\documentclass[a4paper,11pt]{article}
\usepackage{etex}
\usepackage[T1]{fontenc}
\usepackage[letterpaper,tmargin=2.5cm,bmargin=2.5cm,rmargin=2.5cm,lmargin=2.5cm]{geometry}
\usepackage{amssymb,amsmath}
\usepackage{amsthm}
\usepackage{graphicx}
\usepackage{fancyhdr}
\usepackage{fancybox}
\usepackage{shadow}
\usepackage{empheq}
\usepackage{titlesec}
\usepackage{titletoc}
\usepackage{soul, color}
\usepackage{float}
\usepackage{shorttoc}
\usepackage{xcolor}
\usepackage{fancybox}
\usepackage{natbib}
\usepackage{array}
\usepackage{tabularx}
\usepackage{booktabs}
\usepackage{rotating}
\usepackage{bigstrut}
\usepackage{natbib}
\usepackage{aeguill}
\everymath{\displaystyle}
\newtheorem{thm}{Theorem}

\usepackage{changepage}

\usepackage{hyperref}
\def\sym#1{\ifmmode^{#1}\else\(^{#1}\)\fi}
\usepackage{caption}
\usepackage{blindtext}

\usepackage{caption}
\usepackage{subcaption}
\usepackage{mathtools}

\usepackage{graphicx}

\usepackage{enumerate}

\usepackage{hyperref}
\hypersetup{backref,
colorlinks=true, citecolor=blue!50!black}
\usepackage{graphicx}

\usepackage{titlesec}

\usepackage{graphicx}
\usepackage{setspace}

\begin{document}
\title{Fuzzy Difference-in-Discontinuities: Identification Theory and Application to the Affordable Care Act\thanks{An early version of this paper circulated under the title ``Does Obamacare Care? A Fuzzy Difference-in-Discontinuities Approach". The authors would like to thank seminar participants at McMaster University, University of Toronto, University of Sherbrooke, University of Bristol and University of Kent. Comments form Ismael Mourifie, Jeffrey Racine, Christian Gourieroux, Frank Windmeijer, Youngki Shin, Arthur Sweetman, Sarah Smith, Emmanuel Guerre, Miguel Le\'on-Ledesma and Zaki Wahhaj are gratefully acknowledged. All errors or omissions are ours.}}
\author{Hector Galindo-Silva\thanks{Department of Economics
} \\
%EndAName
Universidad Javeriana\\
\and Nibene Habib Som\'e\thanks{Department of Epidemiology and  Biostatistics
} \\
Western University \\
%Institute for Clinical Evaluative Sciences\\
 \and Guy Tchuente \thanks{%
Corresponding author: School of Economics, University of Kent. E-mail: guytchuente@gmail.com. Address: Kennedy Building, Park Wood Road, Canterbury, Kent, CT2 7FS. Tel:+441227827249} \\
%EndAName
University of Kent
}
\date{ April 2021}
\maketitle
\abstract{\footnotesize{ \linespread{0.1}{This paper explores the use of a fuzzy regression discontinuity design where multiple treatments are applied at the threshold. The identification results show that, under the very strong assumption that the change in the probability of treatment at the cutoff is equal across treatments, a difference-in-discontinuities estimator identifies the treatment effect of interest. The point estimates of the treatment effect using a simple fuzzy difference-in-discontinuities design are biased if the change in the probability of a treatment applying at the cutoff differs across treatments. Modifications of the fuzzy difference-in-discontinuities approach that rely on milder assumptions are also proposed. Our results suggest caution is needed when applying before-and-after methods in the presence of fuzzy discontinuities. Using data from the National Health Interview Survey, we apply this new identification strategy to evaluate the causal effect of the Affordable Care Act (ACA) on older Americans' health care access and utilization.  \\
\textbf{Keywords:} Fuzzy Difference-in-Discontinuities, Identification, Regression Discontinuity Design, Affordable Care Act.\\
\textbf{JEL classification:} C13, I12, I13, I18. }}}

\newpage

\section{Introduction}

This paper proposes a method to identify and estimate the partial effect of the treatment of interest when multiple non-mutually-exclusive treatments have been assigned in a fuzzy manner at the same cutoff. We refer to this new approach as a ``fuzzy difference-in-discontinuities" design. This name follows \cite{grembifiscal} and \cite{EggersFreierGrembiNannicini2018}, who propose a ``difference-in-discontinuities'' approach that combines features of regression discontinuity (RD) and difference-in-differences designs. As we will describe below, our methodology generalizes \cite{grembifiscal} and \cite{EggersFreierGrembiNannicini2018}'s results. Our econometric problem can be viewed as a specific case of a general question: how to evaluate the pure effect of a policy intervention in the presence of confounding interventions. The use of non-mutually-exclusive treatments relates our work to the literature on competing risks in survival analysis (see \cite{fine1999proportional} for a description of competing risk models). In survival analysis, a life may end due to one of many risks, similar to how in policy analysis, an outcome can be caused by the policy of interest or a confounding factor. The main difficulty in policy analysis is that the treatment decision is usually endogenous, while the treatment's effects are heterogenous. How confounding policies affect other policy evaluation methods, such as problem structuring methods, difference-in-differences designs, synthetic control matching or instrumental variables, is left for future research. We use this method to identify the causal effect of the Affordable Care Act (ACA) on health care access and utilization for seniors at age 65.\footnote{As we describe below, our empirical application essentially combines \cite{card2008impact}'s fuzzy regression discontinuity design with the difference-in-differences design.}

Our ``fuzzy difference-in-discontinuities" method requires panel data or a pooled cross-sectional sample of the population, where at least one age group is eligible for treatment by all of the policies, while others are eligible for all but the policy of interest. Our identification results show that, under the assumption that the change in the probability of a treatment applying at the cutoff is equal across treatments, a fuzzy difference-in-discontinuities regression identifies the treatment effect of interest. If the treatment probabilities are not equal, a point estimate of the treatment effect using the fuzzy difference-in-discontinuities is biased. For this scenario, we propose alternative estimands of the treatment effect under an alternative set of assumptions. Our identification results cover cases with and without selection at the cutoff and are widely applicable. In general, our results suggest caution is needed when applying before-and-after methods in presence of fuzzy discontinuities.

Our method builds on past findings related to regression discontinuities and the use of before-and-after methods.  We specify a set of conditions under which a fuzzy difference-in-discontinuities estimator identifies a local average treatment effect.  We propose identification results similar to those by \cite{hahn2001identification}, but we generalize them to multiple treatments.  \cite{grembifiscal} and \cite{EggersFreierGrembiNannicini2018} propose and implement a sharp difference-in-discontinuities estimator that exploits ``before-and-after" and discontinuous policy variations (See also \cite{Leonardi2013} and \cite{BenedettoPaolo2018}, who use a difference-in-discontinuities approach.)  We extend these works to the case of fuzzy discontinuities.\footnote{To our knowledge, \cite{Jackson2019} is the only study that has combined a difference-in-differences design with a fuzzy regression discontinuity design, but it does not develop new theory or make explicit the assumptions that underlie this kind of specification.} The potential outcomes framework enables us to clarify the conditions under which a particular treatment effect of interest can be identified when many treatments are applied.\footnote{In this respect, see \cite{gilraine2017multiple}, who estimates the effect of class size on student performance in a sharp discontinuity setup.} Our results show that fuzzy treatment assignment leads to very restrictive identification conditions, and therefore should not be ignored.

In the presence of selection on unobservables near the cutoff, a fuzzy RD design can be understood as identifying a local average treatment effect on the compliers. Therefore, this paper is related to the large and growing literature on instrumental variables estimation with multiple treatments (see for instance \citealt{kirkeboen2016field, kline2016evaluating, hull2018isolateing}).\footnote{Also see \cite{lee2018identifying}, who discuss the identification of a multivalued treatment effect in the presence of multidimensional unobserved heterogeneity.} In a context with mutually exclusive treatments and multiple instruments, \cite{kirkeboen2016field} establish a set of conditions for point identification. In settings where one instrument shifts two treatments or when there are multiple counterfactual treatments, \cite{kline2016evaluating} and \cite{hull2018isolateing} consider the use of covariate-instrument interactions as additional instruments. We complement this literature by assuming that the treatment options under consideration are not necessarily mutually exclusive and may not have additive effects on the outcome.

There is a vast and growing literature evaluating the ACA's effects. Some studies have looked at the effect of a specific aspect of the ACA (e.g. Medicaid expansion) on access to care in particular U.S. states (for instance, \citealt{sommers2016} and \citealt{Courtemanche_ambulance}). \cite{sommers2016} use data from Kentucky, Arkansas and Texas, and a difference-in-differences specification, to assess changes in access to care among low-income adults after two years of ACA implementation. They find that Kentucky's Medicaid program and Arkansas's private option were associated with significant increases in access to primary care among low-income adults.  \cite{Courtemanche_ambulance} confirm that the ACA increased health insurance coverage in states that expanded Medicaid, and also look at the ability of health care service providers to meet demand. Importantly, they find that ambulance response times increased substantially with the implementation of the ACA, which is consistent with a supply-adjustment cost coming from an increase in demand. The coverage gains from the ACA's implementation are well documented. For example, \cite{Cohen2015} show that the ACA has reduced the uninsured rate from 16.0\% in 2010 to 9.1\% in 2015. However, relatively little is known about the effects of the ACA on access to and utilization of health care, despite the fact that the expansion of health insurance coverage was expected to increase the ability of a large proportion of the population to pay for health care services. In a recent review, \cite{Laxmaiah2017} find that access to care seems to have diminished under the ACA.  This paper provides new evidence on how the ACA affects older Americans' utilization of health care services. Our findings suggest that the ACA exacerbated cost barriers to health care for seniors. In 2014 (relative to 2012), more 65-year-olds delayed care due to costs (an increase of 3.6\%),  could not afford to pay for prescription drugs (an increase of 7.0\%),  could not afford to see a specialist (an increase of 7.2\%), and could not have a follow-up treatment (an increase of 5.5\%). Interestingly, the effects of the ACA are heterogenous across ethnicities and education levels.

The remainder of the paper is organized as follows. Section 2 develops our fuzzy difference-in-discontinuities estimator. Section 3 contains the empirical application.  Section 4 discusses the results and concludes.

%%%%%%%%%%%%%%%%%%%%%%%%%%%%%%%%%%%%%%%%%%%%%%%%%%%%%%%%%

%%%%%%%%%%%%%%%%%%%%%%%%%%%%%%%%%%%%%%%%%%%%%%%%%%%%%%%%%

\section{Theory: Identifying and Estimating a Policy Effect in the Presence of Confounding Policies}

Consider a population of $N$ individuals, each born in one of $G$ years. Let $Y_{ic}$ be an outcome (e.g. a health-related indicator),  where $i=1,...,N$ indexes the individuals, and $c=1,...,G$ indexes the yeas.  Define $O_{ic}$ as an indicator variable that identifies whether individual $i$ born in year $c$ is affected  by the policy of interest.

Before the introduction of $O_{ic}$, another policy was in place. Let  $M_{ic}$ be  an indicator variable that identifies whether individual $i$ born in year $c$ participated in this original policy.  The selection of participants in  $M_{ic}$  is partially determined by a forcing variable $X_{ic}$, and changes discontinuously at the cutoff or threshold $t$. Specifically, we say that an individual $i$ born in year  $c$ is treated --- with a higher probability --- when $X_{ic}>t$. The fact that the treatment status is partially determined by a forcing variable $X_{ic}$ means that individuals for whom $X_{ic}<t$ may also be treated by the policy. In this sense, program participation is fuzzy.  The selection of participants in  $O_{ic}$  is only partially determined by  $X_{ic}$ and $t$; it also depends on the year of birth of individual $i$. In this respect, we distinguish between two types of individuals, young and old, denoted by $L$ and $\bar{L}$ respectively, and say that individual $i$ is treated by $O_{ic}$ only if they belong to the younger group, $c \in L$.

Even though we focus on a fuzzy setting, it is useful to describe this assignment mechanism when $O_{ic}$ and $M_{ic}$ are deterministic functions of the running variables. In this case,
\begin{equation}
\label{Oicdef}
  O_{ic}=\left\{ \begin{array}{c}
                1\hspace{0.2 cm} if \hspace{0.2 cm} X_{ic}>t \hspace{0.2 cm}and \hspace{0.2 cm} c\in L \hspace{0.2 cm}\\
                0 \hspace{2 cm} \hspace{0.2 cm} otherwise
            \end{array}\right.
\end{equation}
and
\begin{equation}
\label{Micdef}
  M_{ic}=\left\{ \begin{array}{c}
                1\hspace{0.2 cm} if \hspace{0.2 cm} X_{ic}>t \hspace{0.2 cm}\\
                0 \hspace{0.2 cm} \hspace{0.2 cm} otherwise
            \end{array}\right.
\end{equation}

\bigskip

We define $Y_{ic}(o,m)$ as the potential outcome for individual $i$ born in year $c$ if  $O_{ic}=o$ and $M_{ic}=m$, where $m $, $o$ $\in \{ 0, 1\}$, with 1 corresponding to the individual being treated and 0 otherwise. By (\ref{Oicdef}) and (\ref{Micdef}), the observed outcome is  equal to
\begin{eqnarray}\label{out}
\nonumber
   Y_{ic}&=& O_{ic}M_{ic}Y_{ic}(1,1)+ O_{ic}(1-M_{ic})Y_{ic}(1,0)\\
   &+& (1-O_{ic})M_{ic}Y_{ic}(0,1)+(1-O_{ic})(1-M_{ic})Y_{ic}(0,0)
\end{eqnarray}

Our aim is to identify the causal effect of $O_{ic}$ on $Y_{ic}$. We focus on the average treatment effect of $O_{ic}$ at $t$ for $c \in L$, which we denote by $ATE_O(t)$,  and which we define as
\begin{equation}
\label{LATE}
ATE_O(t)=E(Y_{ic}(1,1)-Y_{ic}(0,1)| X_{ic}=t)
\end{equation}

If $O_{ic}$ is the only treatment   using the cutoff $t$, the cross-sectional  regression discontinuity  estimand would identify the average treatment effect of $O_{ic}$ at $t$. However, in our setting, this estimator will lead to a biased estimate of $ATE_O(t)$ because of the difficulty of separating the effect of $O_{ic}$ from the effect of $M_{ic}$.

Let us define $ATE(t)$ as the cross-sectional  fuzzy regression discontinuity estimand and let $ATE_M(t)$ be the fuzzy regression discontinuity estimand of the effect of $M_{ic}$ without application of $O_{ic}$. In the case of a sharp discontinuity, \cite{grembifiscal} and \cite{EggersFreierGrembiNannicini2018} show that $ATE_O(t)$ can be identified using what they call a difference-in-discontinuities estimand. Specifically, they show that $ATE_O(t)=ATE(t)-ATE_M(t)$. However, we show that, in a fuzzy scenario, this result often does not hold without additional assumptions. As described by \cite{lee2010regression}, in many settings of economic interest, the cutoff only partly determines the treatment status. It is, therefore, possible that the change in the probability of participation differs over time and across different policies.%\footnote{The study of the case in which take-up by program participants is imperfect, and scenarios in which other factors (observable or unobservable) affect the probability of program participation, is crucial for practitioners.}

In the following section, we investigate assumptions under which the difference in the fuzzy discontinuities identifies a policy-relevant quantity when multiple treatments are applied at the same cutoff.  Our theoretical framework follows \cite{hahn2001identification}'s model, extending it to multiple treatments using panel or pooled cross-sectional data. The theoretical discussion on identification considers, as a natural departure point from \cite{grembifiscal}'s and \cite{EggersFreierGrembiNannicini2018}'s difference-in-discontinuities estimators, a fuzzy counterpart. We first assume that there is no selection on unobservables, but include the possibility of heterogeneous treatments. This allows us to focus on the importance of the changes in the proportion of individuals affected by the treatment at the cutoff.  We further relax the assumption of no selection on unobservables, which might be more realistic. By allowing for selection on unobservables, the causal parameter of interest becomes a local average treatment effect and, as previously mentioned, the results are related to recent developments in the estimation with instrumental variables with multiple alternatives.

%%%%%%%%%%%%%%%%%%%%%%%%%%%%%%%%%%

\subsection{Fuzzy Difference-in-Discontinuities: Identification}

Let  $Z_{ic}$ be a random variable, and define the limits $Z^+$,  $Z^-$ and $Z$ as $Z^+=\lim_{x\rightarrow t^+}E[Z_{ic}|X_{ic}=x]$,  $Z^-=\lim_{x\rightarrow t^-}E[Z_{ic}|X_{ic}=x]$, and $Z=\lim_{x\rightarrow t}E[Z_{ic}|X_{ic}=x]$. For any $Z_{ic}$, also define $\tilde{Z}_{ic}= 1\{c\in L\} Z_{ic}$ and  $\bar{Z}_{ic}= 1\{c\in \bar{L}\} Z_{ic}$.

\medskip

To identify the marginal causal effect of $O_{ic}$, we consider the following estimand
\begin{equation}\label{fuzzdiff-disc}
\tau_{O}^{FRD}=\frac{\tilde{Y}^+-\tilde{Y}^-}{\tilde{T}^+-\tilde{T}^-}- \frac{\bar{Y}^+-\bar{Y}^-}{\bar{M}^+-\bar{M}^-}
\end{equation}
where $T_{ic}= O_{ic}M_{ic}$.

\medskip
We call  $\tau_{O}^{FRD}$  in (\ref{fuzzdiff-disc}) a ``fuzzy difference-in-discontinuities'' estimator because, like \cite{grembifiscal}'s and \cite{EggersFreierGrembiNannicini2018}'s estimators, it rests on the intuition of combining difference-in-differences and RD strategies, but in our setting, the RD design is fuzzy. The choice of this estimand is motivated as a simple natural extension of \cite{grembifiscal}'s estimand.

\medskip

In this section, we provide a set of assumptions under which $\tau_{O}^{FRD}$, as defined in (\ref{fuzzdiff-disc}), identifies the $ATE_O(t)$ in (\ref{LATE}). All the assumptions will be conditional on $X_{ic}$  being in the neighborhood of the cutoff $t$.

\bigskip

\textbf{Assumption 1.} The conditional expectation of  each potential outcome  is continuous in $x$  at $t$, i.e., $E[Y_{ic}(o, m)|X_{ic}=x]$ is continuous in $x$ for all $c$ and all $o, m \in \{ 0,1\}$.

\medskip

This first assumption is standard in the RD literature, and states that the conditional expectation of  all potential outcomes is continuous at the cutoff point.

\bigskip

\textbf{Assumption 2.} $M_{ic}$ and $O_{ic}$ are independent of $Y_{ic}(o,m)$, where  $o,m=0,1$.

\medskip

This second assumption states that the determination of whether an individual is subject to the treatment is independent of the potential outcomes near the cutoff, i.e. individuals cannot self-select into the treatment based on their expected benefits. This assumption will be relaxed later to allow for some self-selection. As we are departing from \cite{grembifiscal}'s estimand, this assumption is a natural first step.

\bigskip

\textbf{Assumption 3.} The effect of the confounding policy $M_{ic}$ when there is no treatment
($O_{ic}=0$) is constant across year of birth: $Y_{c_1} (0, 1) - Y_{c_1} (0, 0) = Y_{c_2}(0, 1) - Y_{c_2}(0, 0)$ for any $c_2 \in L$ and $c_1 \in \bar{L}=C\smallsetminus L$, where $C$ is the set of all birth years in the sample.

\medskip

In Assumption 3, the confounding policy must have the same effect before and after the treatment of interest. This assumption can be tested by investigating the treatment effect of several consecutive periods with only one treatment at the cutoff or by comparing the groups that did not receive the treatment before and after the treatment period.

\bigskip

\textbf{Assumption 4.} (i) The limits  $O^+=\lim_{x\rightarrow t^+}E[O_{ic}|X_{ic}=x]$, $O^-=\lim_{x\rightarrow t^-}E[O_{ic}|X_{ic}=x]$,  $M^+=\lim_{x\rightarrow t^+}E[M_{ic}|X_{ic}=x]$, $M^-=\lim_{x\rightarrow t^-}E[M_{ic}|X_{ic}=x]$, $T^+=\lim_{x\rightarrow t^+}E[T_{ic}|X_{ic}=x]$ and $T^-=\lim_{x\rightarrow t^-}E[T_{ic}|X_{ic}=x]$ exist\\ (ii) $O^+\neq O^{-}$, $M^+\neq M^{-}$  and $T^+\neq T^-$.

\medskip

Assumptions 4 (i) and 4 (ii)  are standard RD assumptions for the two policies.

\bigskip

\textbf{Assumption 5.} The discontinuity in the probability of the treatment applying is the same for all policies at the threshold, i.e.
$\tilde{O}^+- \tilde{O}^-=\tilde{T}^+-\tilde{T}^-=\tilde{M}^+ -\tilde{M}^{-}.$

\bigskip

Assumption 5 is new, and is one of the contributions of this paper. It requires that the discontinuity in the probability of selection of each policy be the same as well as the joint probability of selection. This assumption is clearly satisfied when the discontinuity is sharp.

\medskip
The following theorem gives conditions for the identification of the treatment of interest.

\begin{thm}(Identification of the fuzzy difference-in-discontinuities estimator): If Assumptions 1 to 5 hold, then the fuzzy difference-in-discontinuities estimator $\tau_{O}^{FRD}$ defined in (\ref{fuzzdiff-disc}) identifies the average treatment effect, $ATE_O(t)$, in (\ref{LATE}).
\end{thm}

\begin{proof}[\textbf{Proof}]

From Assumptions 1 and 2, first note that
\begin{eqnarray}\nonumber
  \tilde{Y}^+-\tilde{Y}^- &=& \lim_{x\rightarrow t^+}E[\tilde{Y}_{ic}|X_{ic}=x]-\lim_{x\rightarrow t^-}E[\tilde{Y}_{ic}|X_{ic}=x] \\\nonumber
   &=& (\tilde{T}^+-\tilde{T}^-)(\tilde{Y}(1,1)-\tilde{Y}(0,1)) + (\tilde{O}^+-\tilde{O}^-)(\tilde{Y}(1,0)-\tilde{Y}(0,0))\\
   &+&  (\tilde{M}^+-\tilde{M}^-)(\tilde{Y}(0,1)-\tilde{Y}(0,0))-(\tilde{T}^+-\tilde{T}^-)(\tilde{Y}(1,0)-\tilde{Y}(0,0))
\end{eqnarray}
and
\begin{eqnarray}\label{mfp}
\nonumber
  \bar{Y}^+-\bar{Y}^- &=& \lim_{x\rightarrow t^+}E[\bar{Y}_{ic}|X_{ic}=x, ]-\lim_{x\rightarrow t^-}E[\bar{Y}_{ic}|X_{ic}=x] \\
   &=& (\bar{M}^+-\bar{M}^-)(\bar{Y}(0,1)-\bar{Y}(0,0))
\end{eqnarray}

Applying Assumption 3 to equation (\ref{mfp}) and dividing each of the previous equations by $\tilde{T}^+-\tilde{T}^-$ and $\bar{M}^+-\bar{M}^-$, we have
\begin{eqnarray}\nonumber
  \tau_{O}^{FRD} &=& \frac{\tilde{Y}^+-\tilde{Y}^-}{\tilde{T}^+-\tilde{T}^-} - \frac{\bar{Y}^+-\bar{Y}^-}{\bar{M}^+-\bar{M}^-} \\ \label{idenfuzzdiff-disc}
   &=& ATE_O(t)-[1-\frac{\tilde{O}^+-\tilde{O}^-}{\tilde{T}^+-\tilde{T}^-}](\tilde{Y}(1,0)-\tilde{Y}(0,0))\\ \nonumber
   &-&[1- \frac{\tilde{M}^+-\tilde{M}^-}{\tilde{T}^+-\tilde{T}^-}](Y(0,1)-Y(0,0))
\end{eqnarray}
Under Assumptions 1 to 5, the right-hand side of (\ref{idenfuzzdiff-disc}) becomes $ ATE_O(t)$. This means that the fuzzy difference-in-discontinuities estimator identifies the local causal effect of the treatment.
\end{proof}

Note that the proof of Theorem 1 consists of two steps: first, Assumptions 1 to 4 lead to the difference-in-discontinuity expression in equation (\ref{idenfuzzdiff-disc}); then, when Assumption 5 is applied, all the terms other than $ATE_O(t)$ are cancelled out.

\medskip

Theorem 1 provides conditions allowing us to identify the causal effect of the treatment of interest. Assumption 5, while being strong, is a testable assumption: the three terms to which Assumption 5 imposes a strict equality represent the discontinuities in program participation at the threshold. Theorem 1 can be viewed as a negative result because it shows that the simple extension of the \cite{grembifiscal} difference-in-discontinuities estimand to the fuzzy case identifies the treatment of interest only under very restrictive assumptions.

In empirical applications, Assumptions 1 to 4 (i) and (ii) can be easily satisfied. As previously mentioned, these assumptions are similar to those used in a standard RD design. However, the Assumption 5 double equality is a strong assumption. The following assumption relaxes it, and imposes a sort of inclusion of the confounding treatment in the treatment of interest.

\bigskip

\textbf{Assumption 4'.} (i) The limits  $O^+=\lim_{x\rightarrow t^+}E[O_{ic}|X_{ic}=x]$, $O^-=\lim_{x\rightarrow t^-}E[O_{ic}|X_{ic}=x]$,  $M^+=\lim_{x\rightarrow t^+}E[M_{ic}|X_{ic}=x]$  and $M^-=\lim_{x\rightarrow t^-}E[M_{ic}|X_{ic}=x]$ exist;\\ (ii) $O^+\neq O^{-}$ and $M^+\neq M^{-}$; and \\ (iii) It is almost certain that $O_{ic}\geq M_{ic}$ and $\tilde{O}^+- \tilde{O}^-=\tilde{M}^+ -\tilde{M}^{-}= \bar{M}^+-\bar{M}^-$.

\medskip

The following theorem gives an alternative set of conditions under which our fuzzy difference-in-discontinuities estimator identifies the treatment effect of interest.

\begin{thm}(Less restrictive identification of the fuzzy difference-in-discontinuities estimator): If Assumptions 1 to 3 and 4' hold, then the fuzzy difference-in-discontinuities estimator $\tau_{O}^{FRD}$ defined in (\ref{fuzzdiff-disc}) identifies the  average treatment effect, $ATE_O(t)$, in (\ref{LATE}).
\end{thm}
\begin{proof}[\textbf{Proof}]
Note that under Assumption 4 (iii), $E[M_{ic}|X_{ic}=x]=P(M_{ic}=1|X_{ic}=x)= P(M_{ic}=1, O_{ic}=1|X_{ic}=x)+P(M_{ic}=1, O_{ic}=0|X_{ic}=x)$.
Hence, $E[M_{ic}|X_{ic}=x]=P(M_{ic}=1|X_{ic}=x)= P(M_{ic}=1, O_{ic}=1|X_{ic}=x)$ given that $O_{ic}\geq M_{ic}$. Thus, $E[M_{ic}|X_{ic}=x]=P(M_{ic}=1|X_{ic}=x)= P(M_{ic}=1, O_{ic}=1|X_{ic}=x)=E[T_{ic}| X_{ic}=x]$. This implies that $\tilde{O}^+- \tilde{O}^-=\tilde{M}^+ -\tilde{M}^{-}$ is enough for Assumption 5 to be verified.

\end{proof}

The assumptions under which Theorem 2 holds are slightly less restrictive than those for Theorem 1.  Moreover, and importantly, the restrictions $O_{ic}\geq M_{ic}$  and $\tilde{O}^+- \tilde{O}^-=\tilde{M}^+ -\tilde{M}^{-}$  are empirically testable.  In our empirical application, these two relations together imply that the change in the participation probability in the treatments as a result of being older than 65 should be constant.  The strict equality in Assumption 4' (iii) is still very strong (even though it is less restrictive than Assumption 5), as it means that in the case of strict inclusion, the difference on both sides of the cutoff should be similar (i.e. that $\tilde{O}^+-\tilde{M}^+= \tilde{O}^--\tilde{M}^{-}$). If there is selection on unobservables, the assumption may not hold.  The following assumption provides another alternative to Assumption 5.

\bigskip

\textbf{Assumption 5'.} The  two non-mutually-exclusive treatments interact in an additive manner, i.e.  $\tilde{Y}(1,1)-\tilde{Y}(0,1)=\tilde{Y}(1,0)-\tilde{Y}(0,0)$.

\medskip

In Assumption 5', we also assume that the effect of the second treatment would have been the same with or without the confounding treatment. This assumption may not be empirically testable. Non-mutually-exclusive treatments could amplify or mitigate the effect of the treatment of interest. Nevertheless, using this untestable assumption, we are able to relax the equality assumption.

  \begin{thm}(Identification of the $ATE_O$ additive treatment): Under Assumptions 1 to 3, 4 (i), 4 (ii) and 5',
   \begin{enumerate}
     \item [a)] $\frac{\tilde{T}^+-\tilde{T}^-}{\tilde{O}^+-\tilde{O}^-}\left[\tau_{O}^{FRD}+ [1- \frac{\tilde{M}^+-\tilde{M}^-}{\tilde{T}^+-\tilde{T}^-}]ATE_M(t)\right]= \frac{\tilde{Y}^+-\tilde{Y}^-}{\tilde{O}^+-\tilde{O}^-} -\frac{\tilde{M}^+-\tilde{M}^-}{\tilde{O}^+-\tilde{O}^-}ATE_M(t)$ point identifies the $ATE_O(t)$; and
     \item [b)] If it is almost certain that $O_{ic}\geq M_{ic}$, the $ATE_O(t)$  is point identified by $\tau_{O}^{FRD}\frac{\tilde{M}^+-\tilde{M}^-}{\tilde{O}^+-\tilde{O}^-}$.
   \end{enumerate}

\end{thm}
\begin{proof}[\textbf{Proof}]
Theorem 3 shows an alternative way to point identify the treatment effect of interest using a transformation of the difference-in-discontinuities estimator.
As shown in the proof of Theorem 1, when Assumptions 1 to 4 (i) and (ii) are satisfied,
\begin{eqnarray}\label{idenfuzzdiff-disc1}
  \tau_{O}^{FRD} &=& ATE_O(t)-[1-\frac{\tilde{O}^+-\tilde{O}^-}{\tilde{T}^+-\tilde{T}^-}](\tilde{Y}(1,0)-\tilde{Y}(0,0))\\ \nonumber
   &-&[1- \frac{\tilde{M}^+-\tilde{M}^-}{\tilde{T}^+-\tilde{T}^-}](\tilde{Y}(0,1)-\tilde{Y}(0,0))
\end{eqnarray}
Under Assumption 5', i.e. that $O_{ic}$ would have the same treatment effect without $M_{ic}$, we can also say that
\begin{equation}
\label{ateobounds1b}
ATE_O(t)=\frac{\tilde{T}^+-\tilde{T}^-}{\tilde{O}^+-\tilde{O}^-}\left[\tau_{O}^{FRD}+ [1- \frac{\tilde{M}^+-\tilde{M}^-}{\tilde{T}^+-\tilde{T}^-}]ATE_M(t)\right]
\end{equation}

Note that as $O_{ic}\geq M_{ic}$ and $\tilde{Y}(1,1)-\tilde{Y}(0,1)=\tilde{Y}(1,0)-\tilde{Y}(0,0)$, we have that $[1- \frac{\tilde{M}^+-\tilde{M}^-}{\tilde{T}^+-\tilde{T}^-}]=0$, and the result follows from Equation(\ref{ateobounds1b}).
\end{proof}

\bigskip

So far, our set and point identification results have assumed that there was no selection based on potential outcomes (i.e. Assumption 2). The following assumption allows us to relax Assumption 2, generalizing our identification results to scenarios with selection on unobservables.

\bigskip

\textbf{Assumption 6.} (i) $(Y_{ic}(o,m)-Y_{ic}(o_1,m_1), O_{ic}(x))$ and $(Y_{ic}(o,m)-Y_{ic}(o_1,m_1), M_{ic}(x))$ are jointly independent of $X_{ic}$ near the cutoff $t$, with $m, m_1, o, o_1 \in \{0,1\}$  and  $ O_{ic}(x)$ and $M_{ic}(x)$ are treatment states given $X_{ic}=x$. \\ (ii) There exists an $\varepsilon>0$ such that $O_{ic}(t+e)\geq O_{ic}(t-e)$, $M_{ic}(t+e)\geq M_{ic}(t-e)$ and $T_{ic}(t+e)\geq T_{ic}(t-e)$ for all $0<e<\varepsilon$.\\ (iii) There exists an $\varepsilon>0$  such that if $e>0$ and is sufficiently small (i.e. $0<e<\varepsilon$),\\ $E[Y_{ic_1}(0,1)-Y_{ic_1}(0,0)| \{M_{ic_1}(t+e)- M_{ic_1}(t-e)=1\}]=E[Y_{ic_2}(0,1)-Y_{ic_2}(0,0)| \{M_{ic_2}(t+e)- M_{ic_2}(t-e)=1\}]$  for any $c_2 \in L$ and $c_1 \in \bar{L}=C\smallsetminus L$.

\medskip

Assumption 6 (i) means that the choice of the cutoff is exogenous. It allows for selection based on potential outcomes. Assumption 6 (ii)  is similar to the monotonicity assumption in the instrumental variables literature. Assumption 6 (iii) is the analogue of Assumption 3 when there is selection on unobservables.

\medskip

\begin{thm}(Local average treatment effect for the fuzzy difference-in-discontinuities model): Suppose that Assumptions 1, 4, 5 and 6 hold.   Then, $\tau_{O}^{FRD}$  identifies a local average treatment effect, i.e.,
\begin{equation}\label{Ladifdis}
   \tau_{O}^{FRD}= \lim_{e\rightarrow 0}E(Y_{ic}(1,1)-Y_{ic}(0,1)| \{O_{ic}(t+e)- O_{ic}(t-e)=1\},\{M_{ic}(t+e)- M_{ic}(t-e)=1\})
 \end{equation}
\end{thm}
\begin{proof}[\textbf{Proof}]
Let us consider the following quantity $A$, evaluated for $c \in L$:
\begin{eqnarray}\nonumber
A &=&  E[Y_{ic}|X_{ic}=t+e]-E[Y_{ic}|X_{ic}=t-e] \\\nonumber
&=& E[O_{ic}M_{ic}Y_{ic}(1,1)+ O_{ic}(1-M_{ic})Y_{ic}(1,0)|X_{ic}=t+e]\\\nonumber
&+&  E[(1-O_{ic})M_{ic}Y_{ic}(0,1)+(1-O_{ic})(1-M_{ic})Y_{ic}(0,0)|X_{ic}=t+e]\\\nonumber
&-&E[O_{ic}M_{ic}Y_{ic}(1,1)+ O_{ic}(1-M_{ic})Y_{ic}(1,0)|X_{ic}=t-e]\\
&-&E[(1-O_{ic})M_{ic}Y_{ic}(0,1)+(1-O_{ic})(1-M_{ic})Y_{ic}(0,0)|X_{ic}=t-e]
\end{eqnarray}
From the independence assumption (Assumption 6 (i)) and monotonicity assumption (Assumption 6 (ii)), which are similar to arguments in \cite{hahn2001identification}, Theorem 3, the last expression of $A$ is equivalent to
 \begin{eqnarray}
A &=&  E[Y_{ic}(1,1)-Y_{ic}(0,1)| \{O_{ic}(t+e)- O_{ic}(t-e)=1\}, \{M_{ic}(t+e)- M_{ic}(t-e)=1\}]\\\nonumber
 &\times& (E[T_{ic}|X_{ic}=t+e]-E[T_{ic}|X_{ic}=t-e])\\\nonumber
 &+& E[Y_{ic}(1,0)-Y_{ic}(0,0)| \{O_{ic}(t+e)- O_{ic}(t-e)=1\}](E[O_{ic}|X_{ic}=t+e]-E[O_{ic}|X_{ic}=t-e])\\\nonumber
 &+& E[Y_{ic}(0,1)-Y_{ic}(0,0)| \{M_{ic}(t+e)- M_{ic}(t-e)=1\}](E[M_{ic}|X_{ic}=t+e]-E[M_{ic}|X_{ic}=t-e])\\\nonumber
   &-&   E[Y_{ic}(1,0)-Y_{ic}(0,0)| \{O_{ic}(t+e)- O_{ic}(t-e)=1\}](E[T_{ic}|X_{ic}=t+e]-E[T_{ic}|X_{ic}=t-e])
\end{eqnarray}
Applying a similar argument to the older group, we also have that
\begin{eqnarray}\nonumber
  B &=& E[\bar{Y}_{ic}|X_{ic}=t+e]-E[\bar{Y}_{ic}|X_{ic}=t-e] \\\nonumber
   &=&  E[Y_{ic}(0,1)-Y_{ic}(0,0)| \{M_{ic}(t+e)- M_{ic}(t-e)=1\},c \in \bar{L}]\\
   &\times& E[M_{ic}|X_{ic}=t+e,c \in \bar{L})-E(M_{ic}|X_{ic}=t-e, c \in \bar{L}]
\end{eqnarray}
Under Assumptions 3 and 6, we have that $E[Y_{ic}(0,1)-Y_{ic}(0,0)| \{M_{ic}(t+e)- M_{ic}(t-e)=1\},c \in \bar{L}]=E[Y_{ic}(0,1)-Y_{ic}(0,0)| \{M_{ic}(t+e)- M_{ic}(t-e)=1\}]$ for all $0<e<\varepsilon$.  In addition, dividing $A$ by $E[T_{ic}|X_{ic}=t+e]-E[T_{ic}|X_{ic}=t-e]$ and $B$ by $E[M_{ic}|X_{ic}=t+e,c \in \bar{L})-E(M_{ic}|X_{ic}=t-e, c \in \bar{L}]$,  letting $e$ go to zero and applying Assumption 5, we obtain (\ref{Ladifdis}).
\end{proof}

It is important to note that under Assumptions 1, 4', and 6, it can be shown that the fuzzy difference-in-discontinuities model identifies the marginal local average treatment effect (LATE) of the policy or treatment of interest. Moreover, a transformation similar to that obtained in Theorem 2 will also point identify the LATE of the second treatment.

We have shown that a difference-in-discontinuities design can help separate the effects of a policy of interest from those of confounding treatments. We are interested in cases where the treatment is not mutually exclusive and may affect the outcome in a non-additive manner. Identification can be achieved even when there is selection at the threshold based on the potential benefits of a policy. Theorem 4 shows that our fuzzy difference-in-discontinuities estimator identifies the LATE at the discontinuity point.

%In our empirical application, there is a difference between eligibility and participation, since the choice of enrolling in Medicare before or after the ACA can be driven by factors that are unobservable to econometricians but known to the agent. Therefore, our estimated causal effect can be best described as a LATE. The set of compliers is formed by the elderly, whose decision to use Medicare or the ACA's version of Medicare is driven by age-related eligibility criteria. Moreover, and importantly, the ACA and Medicare are not mutually exclusive (and could be view as complements); thus, a traditional instrumental variables approach may not be appropriate.

The identification results presented in this section show conditions for point identification of the $ATE_O(t)$. We have shown that point identification can occur in two  scenarios. First,  the changes in the  treatment probability for both treatments as well as their joint probabilities are equal at the cutoff point. Alternatively, the joint probabilities of treatments might not be needed, as long as the pre-existing treatment is included in the treatment of interest when its application starts (this corresponds to empirical situations where the second treatment is a reinforcement of the existing one). Additionally, we can relax the assumption of equality of treatment probability changes at the cutoff point, replacing it with the assumption that treatment effects are additive  (i.e. that $Y(1,1)-Y(0,1)=Y(1,0)-Y(0,0)$). However, this assumption may not be testable.

In all these cases, point identification using a fuzzy difference-in-discontinuities approach relies on strong testable assumptions. In the case of strict inclusion of the pre-existing treatment in the treatment of interest, the assumption of equality of treatment probability changes at the cutoff point means that the difference should stay exactly the same above and below the cutoff. When the equality of treatment probability changes assumption is relaxed, an additivity assumption is required for point identification, ruling out the case of strict superadditivity  ($Y(1,1)-Y(0,1) > Y(1,0)-Y(0,0)$) or subadditivity ($Y(1,1)-Y(0,1)<Y(1,0)-Y(0,0)$), and the estimator used is not a direct difference-in-discontinuities estimator.

%%%%%%%%%%%%%%%%%%%%%%%%%%%%%%%%%%

\subsection{Estimation and Inference}

The estimation and inference of the treatment effect of interest (i.e. of  $\tau_{O}^{FRD}$ in last section) can be done using a nonparametric approach. In this subsection, we present the steps of the nonparametric procedure. 

The estimation of the treatment effect of interest is obtained using a fuzzy difference-in-discontinuities design via a difference in two ratios. The theorems of the previous section show assumptions under which the difference in the ratios
\begin{equation}\label{parm}
 \tau_{O}^{FRD}= \frac{\tilde{Y}^+-\tilde{Y}^-}{\tilde{T}^+-\tilde{T}^-} - \frac{\bar{Y}^+-\bar{Y}^-}{\bar{M}^+-\bar{M}^-}
\end{equation}
identifies the treatment effect of the relevant policy at $X=t$. Therefore, to obtain a consistent estimator for $\tau_{F}^{FRD}$, we can use consistent estimators of   $\hat{\tilde{Y}}^+$, $\hat{\tilde{Y}}^-$, $ \hat{\tilde{T}}^+$, $\hat{\tilde{T}}^-$
$\hat{\bar{Y}}^+$, $\hat{\bar{Y}}^-$, $ \hat{\bar{M}}^+$ and $\hat{\bar{M}}^-$.

These quantities are commonly estimated using nonparametric regression techniques (see \cite{hahn2001identification}, and \cite{porter2003estimation}, \cite{otsu2015empirical}). The parameters  can be estimated by local linear regression estimators, which are optimal (see for instance  \citealt{porter2003estimation}) and have better boundary properties than traditional kernel regressions (for example, see \citealt{fan1992design}).

The estimator for $\tilde{Y}^+$ is given by a solution to the following weighted least squares problem, where $\hat{\tilde{Y}}^+=\hat{a}$:
\begin{equation}\label{LLRy+}
  (\hat{a},\hat{b})=argmin_{a,b}\sum_{i,c\in L :X_{ic}\geq t} (Y_{ic}-a-b(X_{ic}-t))^2 \mathbb{K}\left(\frac{X_{ic}-t}{h}\right)
\end{equation}
where $ \mathbb{K}$  is the kernel function and $h=h_N$ is the bandwidth satisfying $h\rightarrow 0$ as $N\rightarrow \infty$.

The other quantities included in the first ratio on the right of (\ref{parm}) are estimated using the same type of procedure as in (\ref{LLRy+}). Depending on the quantity we are interested in, $Y_{ic}$ is replaced by $T_{ic}$ or $M_{ic}$. The minimization is made on $X_{ic}\geq t$ or $X_{ic}\leq t$ to get the upper and lower limit estimators, respectively.  Note that in the estimation of this first ratio, we use individuals from the year of birth to which both policies are applied.

To obtain the treatment effect of our policy of interest, we need an estimate of the second ratio on the right side of (\ref{parm}). To estimate the terms comprising this second ratio, we follow a similar procedure to that applied to the elements of the first ratio, but with one difference: the sample now consists of those individuals in the age group to which only one policy (the confounding policy) is applied.
For instance, the estimator for $\bar{Y}^+$ solves the following weighted least
squares problems with respect to $a$, i.e. $\hat{\bar{Y}}^+=\hat{a}$:
\begin{equation}\label{LLRybar}
  (\hat{a},\hat{b})=argmin_{a,b}\sum_{i,c\in \bar{L}:X_{ic}\geq t} (Y_{ic}-a-b(X_{ic}-t))^2 \mathbb{K}\left(\frac{X_{ic}-t}{h}\right).
\end{equation}

The use of two independent samples to evaluate the two ratios ensures the independence of these two quantities. Following Theorem 4 of \cite{hahn2001identification}, the  asymptotic distribution of  the estimator is normally distributed, with its mean given by the difference in means of the two ratios, and the variance given by the sum of the variances. The speed of convergence is $n^{\frac{2}{5}}$, and $h=O_p(n^{-\frac{1}{5}})$ where $n=min(N_1, N_2)$ ($N_1$ is the number of individuals in $P$ and $N_2$ is the number of individuals in $\bar{L}$). The asymptotic results can be  established with a balanced sample in the two age groups. If the samples are not balanced, we can drop the excess randomly. The conventional Wald-type confidence set for $\tau_{O}^{FRD}$ can be obtained by estimating  asymptotic variances of the non-parametric estimator, or by using an appropriate bootstrap method. Another alternative may be to use the empirical likelihood-based inference methods proposed by \cite{otsu2015empirical}, which circumvent the asymptotic variance estimation issues and have data-determined shapes. However, the procedure needs to be extended to account for a potentially heteroscedastic panel data set.

This non-parametric approach is implemented by selecting a smoothing parameter, $h$. For a standard regression discontinuity design, this parameter can be optimally chosen using data-driven selection methods (see \citealt{imbens2012optimal} and \citealt{calonico2014robust}). In the case of a fuzzy discontinuity, \cite{imbens2012optimal} suggest proceeding as in \cite{imbens2008regression} by estimating two optimal bandwidths: one for the main regression outcome and a second for the treatment. To apply this recommendation to our case, we must select four optimal bandwidths. The selection of these bandwidths are theoretically based on homoscedasticity assumptions that may not hold for the pooled cross-section data we are using. While a set of bandwidths might be optimal in the sense of minimizing the integrated mean-squared error, its effect on inference is also of interest. Indeed, \cite{calonico2014robust} show that confidence intervals constructed using bandwidths that minimize the integrated  mean-squared errors are not valid. They propose new theory-based, more robust confidence interval estimators for average treatment effects. To our knowledge, no study has generalized this theory to difference-in-discontinuities settings. The generalization of this theory to these settings (sharp and fuzzy) is important and deserves a careful investigation, which we leave for future research.

%%%%%%%%%%%%%%%%%%%%%%%%%%%%%%%%%%%%%%%%%%%%%%%%%%%%%%%%%

\section{Empirical Application: Effect of the ACA}

Our identification theory shows how to tease out the effect of one specific policy when two treatments are applied at the cut-off.  In the case of two distinct policies (none of them is included in the other), Assumption 5 is key for identification. When one of the treatments is included in the other,  Assumption 4’  must hold instead of Assumption 5.   In this section, we applied our identification theory to analyze the effects of ACA on access to health care services.  Enrolling in Medicare before or after 2013 (the date of the application of the ACA) can be driven by factors that are unobservable to econometricians but known to the agent also, the ACA version of Medicare builds on the pre-existing Medicare.   Thus, after 2013 the Medicare offer to the elderly includes the pre-existing Medicare and the new benefits brought by the ACA.  These extra policies are those we would like to assess the effects.   We acknowledge that the lack of distinction between the enrolee of pre-existing Medicare and ACA Medicare version could be viewed as a limitation of this empirical application.   However, the fact that the 2014 Medicare is a combination of the pre-existing Medicare and the changes of ACA implies an inclusion, so Assumption 4' applies and enables the use of our identification theory to extract the net effects of the ACA.

 %For example, the existence of another policy (eg., assumption xx) will be important(?) In the identification of xxx. To mitigate this (these) limitation we tested 
%In our empirical application, there is a difference between eligibility and participation, since the choice of enrolling in Medicare before or after the ACA can be driven by factors that are unobservable to econometricians but known to the agent. Therefore, our estimated causal effect can be best described as a LATE. The set of compliers is formed by the elderly, whose decision to use Medicare or the ACA's version of Medicare is driven by age-related eligibility criteria. Moreover, and importantly, the ACA and Medicare are not mutually exclusive (and could be view as complements); thus, a traditional instrumental variables approach may not be appropriate.

\subsection{Institutional Background}

The ACA brought the most substantial changes to U.S. health care policy since the creation of Medicare and Medicaid in 1965.  These changes were intended to reduce Medicare costs, expand access to health care services, improve quality of care and expand drug coverage. Prior to the ACA, at age 65, people who had worked 40 quarters or more in covered employment were eligible for Medicare, and could also be eligible for Medicaid if their incomes were below a threshold. These eligibility criteria continue under the ACA, but the ACA is more generous for medium-income individuals and slightly more restrictive for high-income seniors.

Medicare (including the ACA's version of Medicare) has four parts. Part A, hospital insurance, provides broad coverage of inpatient expenses including hospital visits, care in skilled nursing facilities, hospice care and home health services. Coverage is free of charge. Part B, medical insurance, covers medical services including physician fees, nursing fees and preventative services. Enrollees pay a modest monthly premium. Part C, Medicare Advantage, is provided by private insurance; it covers the essentials of Part A and Part B benefits, plus urgent and emergency care services. Its monthly premiums vary widely across private insurers.\footnote{See \url{https://www.medicareresources.org/medicare-benefits/medicare-advantage/}.} Part D, prescription drug coverage, was enacted in 2003 to reduce costs, increase efficiency and improve access to prescription medications for seniors and disabled persons.

When the ACA was introduced in 2010, it came with some improvements/changes to Medicare. This included a gradual reduction in the cost of private insurance premiums (Part C): on average, the payment amount per enrollee decreased by about 6\% in 2014.  The ACA has reduced out-of-pocket expenses for medication of Medicare Part D beneficiaries from 100\% of the coverage gap to 50\% in 2011, making prescription drugs more affordable. Moreover, under the ACA, Medicare beneficiaries (of whom there were over 20 million in 2011) have access to free preventative care services. This includes mammograms, prostate cancer screenings, depression screenings, obesity screenings and counseling, diabetes screenings and screenings for heart disease. The ACA introduced an important modification to care providers' compensation systems under Medicare by moving away from a fee-for-service system to a capitation system with some quality requirements. For example, hospitals with high readmission rates now receive lower payments. Moreover, the new payment system includes financial incentives for care providers to report on different quality measures, including measures that account for the patient's experience.

The main ACA coverage provisions had taken effect by 2014 \citep{Obama2016}. Figure 1 in \cite{Obama2016}  shows that the percentage of individuals without insurance in the U.S. substantially dropped in 2014. This is consistent with the results in \cite{sommers2016} and \cite{Courtemanche_ambulance}.

Medicare after the ACA contains the main characteristics of Medicare before ACA, with some additional benefits and changes in the U.S. health care system. We use ``Medicare'' to refer to the pre-existing Medicare program, and consider the additional elements to be a different policy (ACA).

\subsection{Estimation Equations}

We use a linear model to estimate the effect of the ACA on the utilization of health care services by elderly Americans. In addition to sidestepping the theoretical limitations of the non-parametric approach, this specification enables us to compare our results with those of \cite{card2008impact}.

We restrict our attention to linear regression functions using observations distributed within a distance of $10$ years on both sides of the age 65 cutoff, before and after the implementation of the ACA. We also explore robustness to the inclusion of second-order polynomial terms of age along with interactions and the use of a smaller bandwidth. As discussed below, the estimated discontinuities are generally robust.

We aim to estimate the following model:
\begin{equation}
\label{empir_mod}
Y_{ic}= \alpha_1 +  \alpha_2 M_{ic}+ \alpha_3 O_{ic}+\alpha_4 D_c + \tau_{O}^{FRD} D_c  M_{ic} O_{ic}+f(X_{ic},D_c)+\eta_{ic}
\end{equation}
and $ M_{ic}=\tau_0+\tau_1 X_{ic}^*+\tau_2 D_c+\tau_3 D_c X_{ic}^*+f(X_{ic},D_c)+\varsigma_{ic}$ and $O_{ic}= \pi_0+\pi_1 X_{ic}^*+\pi_2 D_c+\pi_3 D_c X^*_{ic}+f(X_{ic},D_c)+\upsilon_{ic}$, where $X_{ic}$  is the age of individual $i$ in year $c$, $X_{ic}^*$  is a dummy equal to one if this individual is above the age-65 threshold,  $D_c$ is an indicator for the post-ACA period, and $f(X_{ic},D_c)$ is a polynomial function of $X_{ic}$ whose terms include interactions with $D_c$.  As the design is not sharp, $M_{ic}$ (participation in Medicare) and $O_{ic}$ (participation in the ACA) are only partly determined by crossing the age-65 cutoff.  Indeed, some individuals are eligible for Medicare before 65 for disability reasons, and being eligible after 65 is contingent on having worked at least 40 quarters in covered employment. The estimator of the coefficient $\tau_{O}^{FRD}$ is our fuzzy difference-in-discontinuities estimator, and we obtain it through a two-stage-least-squares-type estimation.

The model in equation (\ref{empir_mod}) has been specified to reflect the general theoretical framework proposed in the previous section. However, in practice, the implementation of the ACA for individuals at age 65 included an extension of pre-existing Medicare benefits. This means that the model estimated is simpler, given by $$Y_{ic}= \alpha_1 +  \alpha_2 M_{ic}+\alpha_3 D_c+ \delta D_c \times M_{ic}+f(X_{ic},D_c)+\omega_{ic},$$ with $M_{ic}=\tau_0+\tau_1 X_{ic}^*+\tau_2 D_c+\tau_3 D_c\times X_{ic}^*+f(X_{ic},D_c)+\varphi_{ic}$. 

We consider several outcome variables ($Y_{ic}$), all related to health care access or use: whether a person delayed care last year for cost reasons; whether a person did not get care last year for cost reasons; whether a person saw a doctor or went to the hospital last year; whether a person could afford prescription medications, see a specialist, or receive follow-up care last year; and whether a person could get an appointment soon enough last year.

%%%%%%%%%%%%%%%%%%%%%%%%%%%%%%%%%

\subsection{Data}

We use survey data from the National Health Interview Survey (NHIS).\footnote{This data is available at \url{https://www.cdc.gov/nchs/nhis/data-questionnaires-documentation.htm}} In our baseline specification, we focus on 2012 and 2014 because, as previously described, major policy changes occurred in many states in 2013. Thus, for those states in which these changes occurred, 2013 is a reasonable choice for ACA implementation. Then, we take 2012 and 2014 as representing two moments in which crucial components related to the ACA had either been implemented or not.\footnote{Because of restrictions related to the birthdate of people surveyed, we could not include data for 2015 or 2016.}

For 2012 and 2014, the NHIS reports respondents' birth years and birth months, and what quarter of the calendar quarter the survey took place. We use this information to identify the age (rounded to the nearest quarter) of the respondents.  As in \cite{card2008impact},  we assume that a person who reaches his 65th birthday in the interview quarter has an age of 65 years and 0 quarters. Assuming a uniform distribution of interview dates, we can say that about one-half of these people will be 0-6 weeks younger than 65, and one-half will be 0-6 weeks older.

We limit our analysis to people who are over 55 and under 75, and to regions in which most states implemented the ACA by 2014. In classifying regions, we follow the scheme in the public NHIS data. This identifies the four Census Regions (Northeast, Midwest, South and West). In our baseline specification, we limit our analysis to the Northeast, Midwest and West regions, where most states implemented the ACA by 2014 (see the Kaiser Family Foundation, at \url{https://www.kff.org/}; see also \url{https://www.advisory.com/daily-briefing/resources/primers/medicaidmap}). Thus, we exclude the District of Columbia and the following states: AL, AR, DE, FL, GA, KY, LA, MD, MS, NC, OK, SC, TN, TX, VA and WV. Because of data use restrictions, in our main specification we include  a few states that did not implement the ACA by 2014 (i.e. ID, KS, ME, MO, NE, SD, UT, WI and WY) and exclude a few jurisdictions that implemented the ACA by 2014 (i.e. AR, DE, DC, KY, LA, MD and WV). However, our analysis is robust to the exclusion of the Midwest. The final sample size is 25,291 individuals, although some outcomes are only available for a smaller subsample.

%%%%%%%%%%%%%%%%%%%%%%%%%%%%%%%%%

\subsection{Evidence on Assumptions}

As previously argued, our proposal for identifying the ACA's effects can be best described as a LATE, and relies on the application of the LATE version of Theorem 2 (i.e.  Theorem 4). For this Theorem, the relevant Assumptions are 1, 4' and 6. In this section we discuss the plausibility of these assumptions to our case study.

Assumption 1 is consistent with most of the relevant literature, which uses age 65 as a threshold in the U.S. (see \cite{card2008impact} for an exhaustive list).

Participation in both Medicare and the ACA is partially determined by the same 65-year age threshold for eligibility. Figure \ref{inscovsqoverallgraph}  illustrates this by showing the age profiles of health insurance coverage estimated separately for each treatment (2012 for Medicare, plotted with circles, and 2014 for the ACA, plotted with diamonds). The figure shows that for each treatment, there is a significant increase in the coverage rates. This suggests that the age threshold of 65 provides a credible source of exogenous variation in insurance status for both policies. This also means that Assumptions 4' $(i)$ and 4' $(ii)$ are likely to be satisfied.

Figure \ref{inscovsqoverallgraph} also illustrates a second and important relationship between the likelihood that a person is eligible for Medicare and the likelihood they are eligible for the ACA, both at the same age-65 threshold: the rise in the share of coverage rates at age 65 is virtually the same for 2012 and 2014. This provides evidence that the probability of selection into Medicare and the ACA's Medicare are likely the same. This is consistent with the second part of Assumption 4' (iii)  (i.e. that $\tilde{O}^+- \tilde{O}^-=\tilde{M}^+ -\tilde{M}^{-}= \bar{M}^+-\bar{M}^-$).

Table \ref{inscovsqoverall} confirms the results in Figure \ref{inscovsqoverallgraph} by showing the effects of reaching age 65 on the insurance status for Medicare (Panel A) and the ACA (Panel B) on five insurance-related variables: the probability of having Medicare coverage, the probability of having any health insurance coverage, the probability of having private coverage, the probability of having two or more forms of coverage, and the probability that an individual's primary health insurance is a managed care program. Column (1) in Panels A and B shows that reaching age 65 significantly increases the probability of having Medicare in 2012 (Panel A) and 2014 (Panel B) and, importantly, that the increase in both probabilities is the same. Panel C confirms this result by showing the estimates of the  difference-in-discontinuities estimator where the dependent variable is insurance status. Column (1) in Panel C shows that the probability of having Medicare coverage is not affected by the ACA's rollout.  Columns (3) to (5) show that this result holds for the probability of having private coverage, the probability of having two or more forms of coverage, and the probability that an individual's primary health insurance is a managed care program.

Two policies that use the same cutoff are likely to be complements or substitutes. As previously argued, in our scenario the policies seem  to be complement. In this respect, note that all 2014 Medicare users are treated by the ACA's Medicare program. This is consistent with the first part of Assumption 4' (iii) (i.e. that $O_{ic}\geq M_{ic}$).

Assumptions 1 and 6 are difficult to test. However, we propose a set of placebo regressions to evaluate their plausibility. Assumption 6 (iii)  stipulates that in the absence of the ACA program, the effect of Medicare on the utilization of health care services should be the same.  In Tables \ref{placebo1} to \ref{placebo3}, we construct placebo fuzzy difference-in-discontinuities estimates for age groups or regions only affected by Medicare at age 65. The results in Table \ref{placebo1} suggest that if we consider the Midwest and South regions, the difference in the discontinuity in the level of access to care and health service utilization for seniors at age 65 is not different from zero at any conventional statistical level. As for Assumptions 6 (i) and (ii), we will assume that no senior will refuse to enroll in Medicare as a result of turning 65 in 2012 or 2014, and we allow for the decision to use the Medicare treatment to be related to returns. Unfortunately, to the best of our knowledge, there is no test for these assumptions for panel data. We assume that they hold in our setting.

Finally, Table \ref{placebo3} shows fuzzy difference-in-discontinuities estimates for several consecutive groups of years (prior to 2014). All but two of the differences are statistically indistinguishable from zero. These placebo regressions suggest that Medicare has the same effect across age groups in absence of the ACA policy. Overall, the placebo regressions suggest that Assumption 3 is reasonable for our sample.

%%%%%%%%%%%%%%%%%%%%%%%%%%%%%%%%%%%

\subsection{Empirical Results}
\label{section:results}

Panel A of Table \ref{Tab:access1} presents the fuzzy difference-in-discontinuities estimates for the effect of the ACA on access to care and health care services utilization for 65-year-old Americans. We consider three self-reported access to care outcomes from the NHIS questionnaires: $(i)$ ``During the past 12 months, has medical care been delayed for the individual because of worry about the cost?'' (first column) $(ii)$ ``During the past 12 months, was there any time when the individual needed medical care but did not get it because the individual could not afford it?'' (second column) $(iii)$ Did the individual have at least one doctor visit in the 12 months? (third column). In the last column, we report estimated  $\tau_{O}^{FRD}$ values for health care services utilization, specifically individuals' overnight hospital stays in the previous year.

The results show that, overall, individuals who turned 65 in 2014 were 3.6\% more likely to delay care due to costs. However, the estimated fuzzy difference-in-discontinuities coefficients on the other two access to care outcomes in columns (2) and (3) are not significant at the standard levels. These results suggest that the effect of the ACA on cost-related access to care is mixed. With respect to health care service utilization, we note a significant 4.8\% increase in hospitalization rates for 65-year-olds in 2014. Panel B of Table \ref{Tab:access1} shows the effect of the ACA on several cost-related access to care outcomes for individuals at age 65. Overall, the proportion of individuals who reported that they could not afford to pay for prescription drugs, see a specialist or have a follow-up treatment increased by 7.0\%, 7.2\%,  and 5.5\%, respectively.

We also perform a subgroup analysis considering ethnicity and education. These results should be treated with caution because identification assumptions may not hold for some subgroups. Table \ref{Tab:accesset} presents the results by ethnicity. It reveals that for both whites (non-Hispanics) and minorities (blacks or Hispanics), there is no significant ACA effect on access to care or health care services utilization (see the first panel of Table \ref{Tab:accesset}). Interestingly, we observe a clear heterogeneity in the ACA's effects within individuals' ethnicity. The ACA increased the proportion of blacks aged 65 who had seen a doctor the previous year by 36.7\%, and the proportion of whites (non-Hispanic) with a least one hospitalization by 5.1\%, but 15.6\% more Hispanics forewent access to care the previous year for cost-related reasons. Moreover, panel B of Table  \ref{Tab:accesset} shows that the proportion of whites (non-Hispanics) who could not afford prescription drugs, a specialist visit or follow-up care all increased as a consequence of the ACA. The proportion of Hispanics who could not afford prescription drugs also increased by 23.3\%. Panel B in Table \ref{Tab:accessed} shows that the ACA significantly increased the proportion of high-school-dropout seniors who could not afford a specialist visit or follow-up care, compared to more educated seniors. An additional 11.4\% of seniors with a college education could not afford prescription drugs.

The results suggest that, in general, the ACA exacerbated cost-related access barriers for seniors. In 2014, more 65-year-olds delayed care, could not see a specialist, or could not maintain continuity of care due to costs. This might be in part due to the fact that the implementation of the ACA is associated with the increase in Medicare Part B premiums and the reduction of the government's payment per enrollee to private insurance companies. The ACA increased the proportion of seniors who could not afford prescription drugs. This is surprising, since the ACA was set to reduce Medicare Part D enrollees' out-of-pocket expenses. The increase in hospitalization rates might arise from paying physicians under the ACA based on the quality of services provided and penalizing hospitals with high readmission rates.  Interestingly, the ACA significantly improved access to physicians' services for blacks, and increased hospital stays for whites (non-Hispanics). However, under the ACA, more Hispanics were unable to access to care for cost-related reasons.

%%%%%%%%%%%%%%%%%%%%%%%%%%%%%%%%%%%%%%

\subsection{Robustness Checks}

Identifying the effect of the ACA on access to care requires that all other factors that might affect a 65-year-old's access to care trend smoothly \citep{card2008impact}. An example of a confounding factor is an individual's employment status, since 65 is the typical retirement age, and employed older adults have been found to have better health outcomes than unemployed older adults \citep{ Kachan2015}. This may lead to a biased $\tau_{O}^{FRD}$ if employment status had a significant impact on individuals' health outcomes at the discontinuity (age 65) in 2014.

The estimated effects of the discontinuity at age 65 on employment status are presented in Table \ref{Tab:emp}. We consider two employment variables: whether an individual is employed, and whether the individual is a full-time employee. The results show non-significant coefficients, which suggests that there are no discontinuities at 65 in both cases. Figure \ref{fig:discemp}  illustrates the continuity at age 65 for employment. We also perform the same test using different subgroups: ethnicity and education. The results are presented in Tables \ref{Tab:empet} and \ref{Tab:emped}. Again, in all cases the results show no evidence of a discontinuity for employment. We obtain similar results with smaller bandwidths (see Table \ref{Tab:empsmall}). Therefore we rule out employment as a confounding factor when estimating the ACA's effect on access to health care services.

We also check the robustness of the results obtained in the previous section to the inclusion of second-order polynomial terms for age (Tables \ref{tab:agesquared1}- \ref{tab:agesquared3}) and the use of a smaller bandwidth (Tables \ref{tab:smallerband1}-\ref{tab:smallerband3}). Overall, the results are qualitatively similar to those presented in section \ref{section:results}. Finally, to re-estimate the ACA's effects, we split the sample into two parts: respondents who are enrolled in Medicare Part D, and respondents who are enrolled in Medicare Part A, B or C. The results in Tables \ref{tab:partD} and \ref{tab:partABC} show similar patterns for the ACA's effects on access to care.

We implement an additional robustness check by removing individuals who turned 65 in the first half of 2014 from the sample. The motivation for this additional check is that for most of our outcomes, the information we use is from the 12 months prior to when the person was interviewed, and if an individual turned 65 at the beginning of 2014, the effect that we are capturing may be more likely to correspond to an event from 2013 instead of 2014. Tables \ref{access_no65} shows that for those outcomes related to the alternative measures of access to care, the results are similar to those presented in section \ref{section:results}, and for the other outcomes, the results are equal in sign but of lesser statistical significance.

%%%%%%%%%%%%%%%%%%%%%%%%%%%%%%%%%%%%%%%%%%%%%%%%%%%%%%%%%

\section{Discussion and Conclusion }

The ACA has generated significant media attention since 2009. Evaluating its effects on the U.S. health care system is necessary to inform the debate on the importance of the ACA. We develop and apply an identification strategy in a fuzzy difference-in-discontinuities design to tease out the causal effect of the ACA on the U.S. population's access to care. Our identification results rely on the presence of a pooled cross-sectional or panel dataset to which a ``before-and-after" policy evaluation can be applied. The partial effect of the policy of interest --- in our application, the ACA --- can be identified under a strict condition of equality in the treatment probability changes at the cutoff point. For the ACA, this condition is likely to be satisfied for the overall sample. We apply our fuzzy difference-in-discontinuities method to self-reported access to care outcomes, using NHIS data over a three-year period (2012-2014).

Our results show that the ACA had an adverse impact on access to health care services for cost reasons. In particular, under the ACA, the likelihood of delaying care due to cost, and the likelihood of being unable to afford a prescribed drug, a specialist visit, or a follow-up treatment, have increased. These results suggest that an increasing number of seniors (aged 65 or older) reported unmet health care needs because of a lack of financial resources. This should concern policymakers, as people who report unmet health care needs face higher risks of mortality \citep{Alonso1997} and of deterioration in their health status \citep{Okumura2013}.

Two mechanisms might explain why the ACA increased cost-related barriers: it increased the demand for health care services by increasing coverage, and it reduced the supply of health services by replacing a generous, uncritical fee-for-service payment model with a capitation-based model in which care providers are paid a fixed amount for each patient to provide a bundle of pre-determined services.\footnote{ For more details, see the Public Law 111--148--MAR. 23, 2010 at  \url{https://www.gpo.gov/fdsys/pkg/PLAW-111publ148/pdf/PLAW-111publ148.pdf}  (accessed in September 2018).} Note that a fee-for-service scheme motivates providers to increase the quantity of services provided \citep{ MCGUIRE2000461}. In contrast, capitation creates incentives to underprovide services, and may improve the quality of services \citep{ScottA2000}. This suggests that along with facilitating access to insurance coverage, the ACA should have included measures or incentives to increase the supply of health services and prevent the increase of insurers' premiums and beneficiaries' out-of-pocket expenses. Our results also show that the ACA improved hospital stays for patients as a result of moving away from a fee-for-service model to a capitation system, which is designed to reward quality instead of quantity.

Our results for access to health care services for the previously insured population may be capturing short-term effects of the ACA. For example, if the ACA successfully increased the quality of care by providing more preventive services, the number of patients per physician  might decrease over time, reducing the demand for care and the quantity consumed. This in turn could reduce access to care issues in the long term. Though estimating such effects will require long-term panel data, which do not yet exist, our identification and estimation strategy will still be valid.

%%%%%%%%%%%%%%%%%%%%%%%%%%%%%%%%%%%%%%%%%%%%%%%%%%%%%%%%%

\newpage

\bibliographystyle{ecca}
\bibliography{Fuzzybib}

\begin{thebibliography}{30}
\providecommand{\natexlab}[1]{#1}

\bibitem[{Alonso \textit{et~al.}(1997)Alonso, Orfila, Ruigomez, Ferrer and
  Anto}]{Alonso1997}
\textsc{Alonso, J.}, \textsc{Orfila, F.}, \textsc{Ruigomez, A.},
  \textsc{Ferrer, M.} and \textsc{Anto, J.~M.} (1997). Unmet health care needs
  and mortality among spanish elderly. \textit{American Journal of Public
  Health}, \textbf{87}~(3), 365--370, pMID: 9096535.

\bibitem[{Benedetto and Paola(2018)}]{BenedettoPaolo2018}
\textsc{Benedetto, M. A.~D.} and \textsc{Paola, M.~D.} (2018). \textit{{Term
  Limit Extension And Electoral Participation. Evidence From A
  Diff-In-Discontinuities Design At The Local Level In Italy}}. Working Papers
  201802, Universita della Calabria, Dipartimento di Economia, Statistica e
  Finanza "Giovanni Anania" - DESF.

\bibitem[{Calonico \textit{et~al.}(2014)Calonico, Cattaneo and
  Titiunik}]{calonico2014robust}
\textsc{Calonico, S.}, \textsc{Cattaneo, M.~D.} and \textsc{Titiunik, R.}
  (2014). Robust nonparametric confidence intervals for
  regression-discontinuity designs. \textit{Econometrica}, \textbf{82}~(6),
  2295--2326.

\bibitem[{Card \textit{et~al.}(2008)Card, Dobkin and Maestas}]{card2008impact}
\textsc{Card, D.}, \textsc{Dobkin, C.} and \textsc{Maestas, N.} (2008). The
  impact of nearly universal insurance coverage on health care utilization:
  evidence from medicare. \textit{American Economic Review}, \textbf{98}~(5),
  2242--58.

\bibitem[{Cohen \textit{et~al.}(2016)Cohen, Martinez and Zammitti}]{Cohen2015}
\textsc{Cohen, R.}, \textsc{Martinez, M.} and \textsc{Zammitti, E.} (2016).
  Early release of selected estimates based on data from the 2015 national
  health interview survey. \textit{National Center for Health Statistics}.

\bibitem[{Courtemanche \textit{et~al.}(2017)Courtemanche, Friedson, Koller and
  Rees}]{Courtemanche_ambulance}
\textsc{Courtemanche, C.}, \textsc{Friedson, A.}, \textsc{Koller, A.~P.} and
  \textsc{Rees, D.~I.} (2017). \textit{The Affordable Care Act and Ambulance
  Response Times}. Working Paper 23722, National Bureau of Economic Research.

\bibitem[{Eggers \textit{et~al.}(2018)Eggers, Freier, Grembi and
  Nannicini}]{EggersFreierGrembiNannicini2018}
\textsc{Eggers, A.~C.}, \textsc{Freier, R.}, \textsc{Grembi, V.} and
  \textsc{Nannicini, T.} (2018). Regression discontinuity designs based on
  population thresholds: Pitfalls and solutions. \textit{American Journal of
  Political Science}, \textbf{62}~(1), 210--229.

\bibitem[{Fan(1992)}]{fan1992design}
\textsc{Fan, J.} (1992). Design-adaptive nonparametric regression.
  \textit{Journal of the American statistical Association}, \textbf{87}~(420),
  998--1004.

\bibitem[{Fine and Gray(1999)}]{fine1999proportional}
\textsc{Fine, J.~P.} and \textsc{Gray, R.~J.} (1999). A proportional hazards
  model for the subdistribution of a competing risk. \textit{Journal of the
  American statistical association}, \textbf{94}~(446), 496--509.

\bibitem[{Gilraine(2017)}]{gilraine2017multiple}
\textsc{Gilraine, M.} (2017). \textit{Multiple treatments from a single
  discontinuity: An application to class size}. Tech. rep., University of
  Toronto.

\bibitem[{Grembi \textit{et~al.}(2016)Grembi, Nannicini and
  Troiano}]{grembifiscal}
\textsc{Grembi, V.}, \textsc{Nannicini, T.} and \textsc{Troiano, U.} (2016). Do
  fiscal rules matter? \textit{American Economic Journal: Applied Economics},
  \textbf{8}~(3), 1--30.

\bibitem[{Hahn \textit{et~al.}(2001)Hahn, Todd and Van~der
  Klaauw}]{hahn2001identification}
\textsc{Hahn, J.}, \textsc{Todd, P.} and \textsc{Van~der Klaauw, W.} (2001).
  Identification and estimation of treatment effects with a
  regression-discontinuity design. \textit{Econometrica}, \textbf{69}~(1),
  201--209.

\bibitem[{Hull(2018)}]{hull2018isolateing}
\textsc{Hull, P.} (2018). Isolateing: Identifying counterfactual-specific
  treatment effects with cross-stratum comparisons. \textit{Available at SSRN
  2705108}.

\bibitem[{Imbens and Kalyanaraman(2012)}]{imbens2012optimal}
\textsc{Imbens, G.} and \textsc{Kalyanaraman, K.} (2012). Optimal bandwidth
  choice for the regression discontinuity estimator. \textit{The Review of
  Economic Studies}, \textbf{79}~(3), 933--959.

\bibitem[{Imbens and Lemieux(2008)}]{imbens2008regression}
\textsc{Imbens, G.~W.} and \textsc{Lemieux, T.} (2008). Regression
  discontinuity designs: A guide to practice. \textit{Journal of econometrics},
  \textbf{142}~(2), 615--635.

\bibitem[{Jackson(2019+)}]{Jackson2019}
\textsc{Jackson, C.~K.} (2019+). Can introducing single-sex education into
  low-performing schools improve academics, arrests, and teen motherhood?
  \textit{Journal of Human Resources}, \textbf{forthcoming}.

\bibitem[{Kachan \textit{et~al.}(2015)Kachan, Fleming, Christ, Muennig, Prado,
  Tannenbaum, Yang, Caban-M and Lee}]{Kachan2015}
\textsc{Kachan, D.}, \textsc{Fleming, L.~E.}, \textsc{Christ, S.},
  \textsc{Muennig, P.}, \textsc{Prado, G.}, \textsc{Tannenbaum, S.~L.},
  \textsc{Yang, X.}, \textsc{Caban-M, A.~J.} and \textsc{Lee, D.~J.} (2015).
  Health status of older us workers and nonworkers, national health interview
  survey, 1997-2011. \textit{Preventing Chronic Disease: Public Health
  Research, Practice, and Policy}, p.~12.

\bibitem[{Kirkeboen \textit{et~al.}(2016)Kirkeboen, Leuven and
  Mogstad}]{kirkeboen2016field}
\textsc{Kirkeboen, L.~J.}, \textsc{Leuven, E.} and \textsc{Mogstad, M.} (2016).
  Field of study, earnings, and self-selection. \textit{The Quarterly Journal
  of Economics}, \textbf{131}~(3), 1057--1111.

\bibitem[{Kline and Walters(2016)}]{kline2016evaluating}
\textsc{Kline, P.} and \textsc{Walters, C.~R.} (2016). Evaluating public
  programs with close substitutes: The case of head start. \textit{The
  Quarterly Journal of Economics}, \textbf{131}~(4), 1795--1848.

\bibitem[{Lee and Lemieux(2010)}]{lee2010regression}
\textsc{Lee, D.~S.} and \textsc{Lemieux, T.} (2010). Regression discontinuity
  designs in economics. \textit{Journal of economic literature},
  \textbf{48}~(2), 281--355.

\bibitem[{Lee and Salani{\'e}(2018)}]{lee2018identifying}
\textsc{Lee, S.} and \textsc{Salani{\'e}, B.} (2018). Identifying effects of
  multivalued treatments. \textit{Econometrica}, \textbf{86}~(6), 1939--1963.

\bibitem[{Leonardi and Pica(2013)}]{Leonardi2013}
\textsc{Leonardi, M.} and \textsc{Pica, G.} (2013). Who pays for it? the
  heterogeneous wage effects of employment protection legislation. \textit{The
  Economic Journal}, \textbf{123}.

\bibitem[{Manchikanti \textit{et~al.}(2017)Manchikanti, Ii, Benyamin and
  Hirsch}]{Laxmaiah2017}
\textsc{Manchikanti, L.}, \textsc{Ii, S.~H.}, \textsc{Benyamin, R.~M.} and
  \textsc{Hirsch, J.~A.} (2017). A critical analysis of obamacare: Affordable
  care or insurance for many and coverage for few? \textit{Pain Physician},
  \textbf{20}~(3), 111--138.

\bibitem[{Mcguire(2000)}]{MCGUIRE2000461}
\textsc{Mcguire, T.~G.} (2000). Chapter 9 - physician agency. In A.~J. Culyer
  and J.~P. Newhouse (eds.), \textit{Handbook of Health Economics},
  \textit{Handbook of Health Economics}, vol.~1, Elsevier, pp. 461 -- 536.

\bibitem[{Obama(2016)}]{Obama2016}
\textsc{Obama, B.} (2016). United states health care reform progress to date
  and next steps. \textit{JAMA}, \textbf{316}~(5), 525--532.

\bibitem[{Okumura \textit{et~al.}(2013)Okumura, Hersh, Hilton and
  Lotstein}]{Okumura2013}
\textsc{Okumura, M.}, \textsc{Hersh, A.}, \textsc{Hilton, J.} and
  \textsc{Lotstein, D.} (2013). Change in health status and access to care in
  young adults with special health care needs: results from the 2007 national
  survey of adult transition and health. \textit{Journal of adolescent health},
  pp. 413--418.

\bibitem[{Otsu \textit{et~al.}(2015)Otsu, Xu and
  Matsushita}]{otsu2015empirical}
\textsc{Otsu, T.}, \textsc{Xu, K.-L.} and \textsc{Matsushita, Y.} (2015).
  Empirical likelihood for regression discontinuity design. \textit{Journal of
  Econometrics}, \textbf{186}~(1), 94--112.

\bibitem[{Porter(2003)}]{porter2003estimation}
\textsc{Porter, J.} (2003). Estimation in the regression discontinuity model.
  \textit{Unpublished Manuscript, Department of Economics, University of
  Wisconsin at Madison}, pp. 1--66.

\bibitem[{Scott(2000)}]{ScottA2000}
\textsc{Scott, A.} (2000). {Economics of general practice}. In A.~J. Culyer and
  J.~P. Newhouse (eds.), \textit{{Handbook of Health Economics}},
  \textit{Handbook of Health Economics}, vol.~1, ~\textit{22}, Elsevier, pp.
  1175--1200.

\bibitem[{Sommers \textit{et~al.}(2016)Sommers, Blendon, Orav and
  Epstein}]{sommers2016}
\textsc{Sommers, B.}, \textsc{Blendon, R.}, \textsc{Orav, E.} and
  \textsc{Epstein, A.} (2016). Changes in utilization and health among
  low-income adults after medicaid expansion or expanded private insurance.
  \textit{JAMA Internal Medicine}, \textbf{176}~(10), 1501--1509.

\end{thebibliography}

%%%%%%%%%%%%%%%%%%%%%%%%%%%%%%%%%%%%%%%%%%%%%%%%%%
%%%%%%%%%%%%%%%%%%%%%%%%%%%%%%%%%%%%%%%%%%%%%%%%%%

%%%%%%%%%%%%%%%%%%%%%%%%%%%%%%%%%%%%%%%%%%%%%%%%%%
%%%%%%%%%%%%%%%%%%%%%%%%%%%%%%%%%%%%%%%%%%%%%%%%%%%%%%%%%%%%%%
%%%%%%%%%%%%%%%%%%%%%%%%%%%%%%%%%%%%%%%%%%%%%%%%%%%%%%%%%%%%%%
\clearpage
\hbox {}

\begin{figure}
\caption{Coverage by age: 2012 vs. 2014}
\centering
\label{inscovsqoverallgraph}
  \includegraphics[width=15cm]{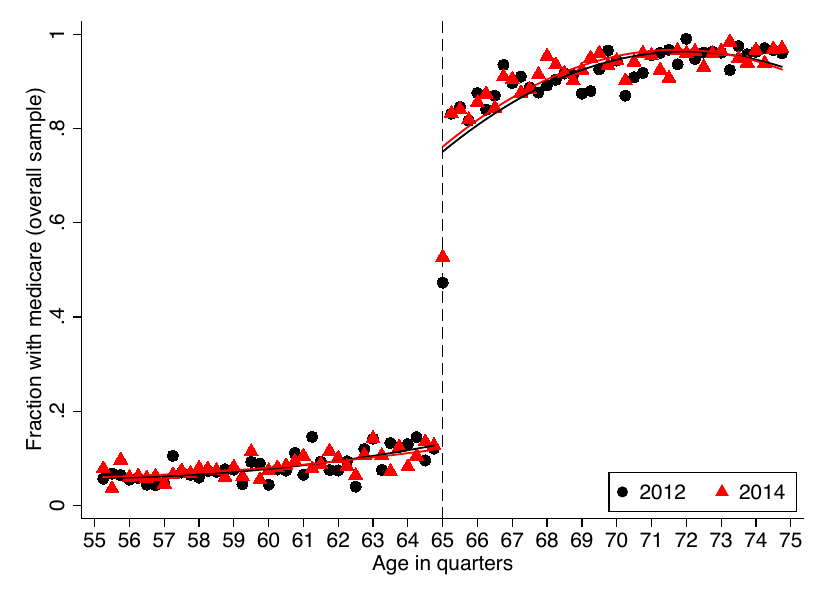}\\
 \label{discconv}
\end{figure}
% In your y-axis heading, "Medicare" should be capitalized.
% Your x-axis heading says "Age in quarters". I would suggest saying "Age (rounded to the nearest quarter of a year)" instead.  If you were actually displaying "Age in quarters", then presumably your cut-off should be at 260 quarters rather than 65 years.

%%%%%%%%%%%%%%%%%%%%%%%%%%%%%%%%%%%%%%%%%%%%%%%%%%%%%%%%%%%%%%

\begin{figure}
\caption{Employment by age: 2012 vs. 2014}
\centering
  \includegraphics[width=15cm]{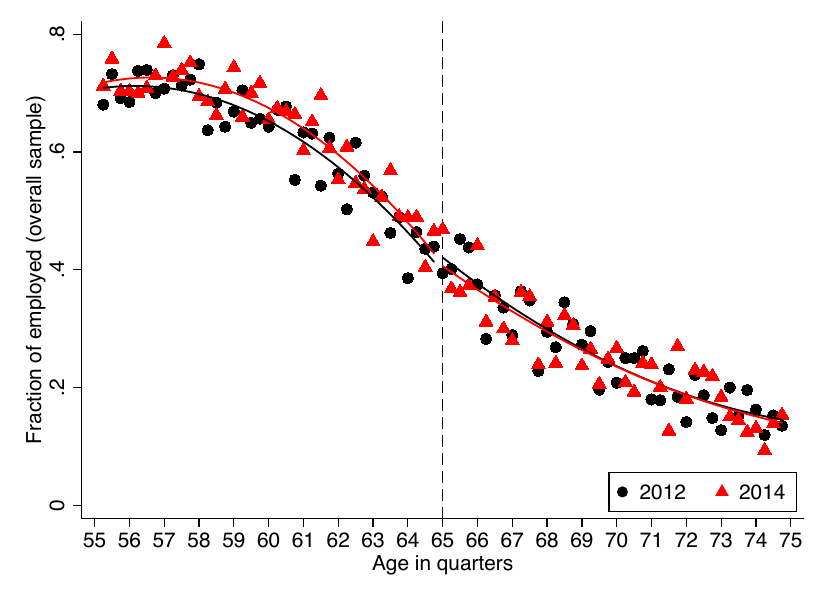}\\
 \label{fig:discemp}
\end{figure}
% Your x-axis heading says "Age in quarters". I would suggest saying "Age (rounded to the nearest quarter of a year)" instead.  If you were actually displaying "Age in quarters", then presumably your cut-off should be at 260 quarters rather than 65 years.

%%%%%%%%%%%%%%%%%%%%%%%%%%%%%%%%%%%%%%%%%%%%%%%%%%%%%%%%%%%%%%

\begin{figure}
\caption{Coverage (Part D) by age: 2012 vs. 2014}
\centering
\label{}
  \includegraphics[width=15cm]{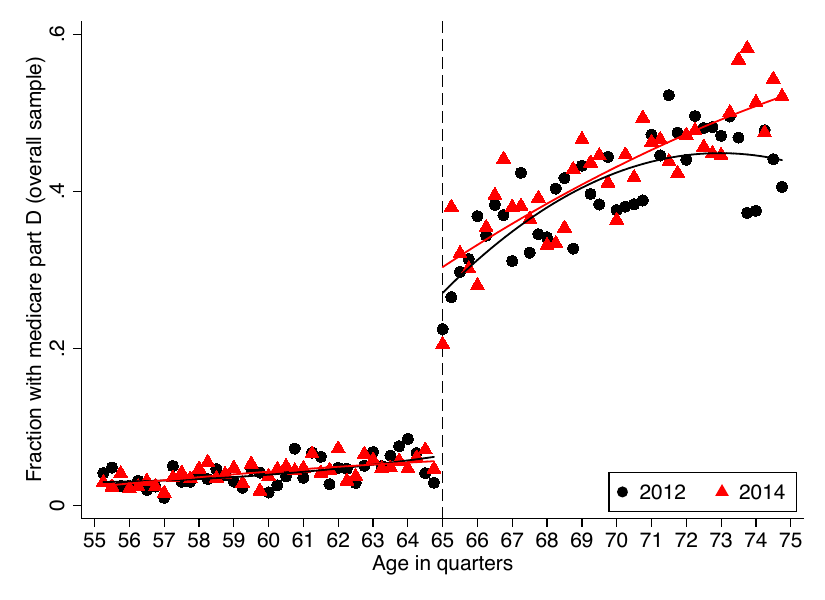}\\
 \label{discconv}
\end{figure}
% In your y-axis heading, "Medicare" should be capitalized.
% Your x-axis heading says "Age in quarters". I would suggest saying "Age (rounded to the nearest quarter of a year)" instead.  If you were actually displaying "Age in quarters", then presumably your cut-off should be at 260 quarters rather than 65 years.

%%%%%%%%%%%%%%%%%%%%%%%%%%%%%%%%%%%%%%%%%%%%%%%%%%%%%%%%%%%%%%

\begin{figure}
\caption{Coverage (Parts A, B or C) by age: 2012 vs. 2014}
\centering
\label{}
  \includegraphics[width=15cm]{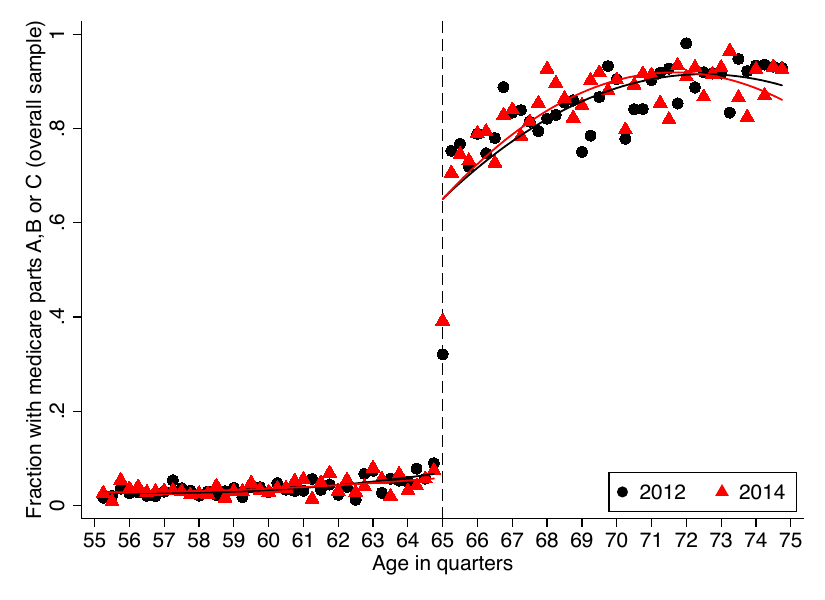}\\
 \label{discconv}
\end{figure}
% In your y-axis heading, "Medicare Parts" should be capitalized.
% Your x-axis heading says "Age in quarters". I would suggest saying "Age (rounded to the nearest quarter of a year)" instead.  If you were actually displaying "Age in quarters", then presumably your cut-off should be at 260 quarters rather than 65 years.

%%%%%%%%%%%%%%%%%%%%%%%%%%%%%%%%%%%%%%%%%%%%%%%%%%%%%%%%%%%%%%
%%%%%%%%%%%%%%%%%%%%%%%%%%%%%%%%%%%%%%%%%%%%%%%%%%%%%%%%%%%%%%

\clearpage
\hbox {}

\begin{table}
\begin{center}
{
\renewcommand{\arraystretch}{0.7}
\setlength{\tabcolsep}{5pt}
\captionsetup{font={normalsize,bf}}
\caption {Insurance Coverage}  \label{inscovsqoverall}
\footnotesize
\centering  \begin{tabular}{lccccc}
\hline\hline
& &   \\
    & \multicolumn{1}{c}{On} & \multicolumn{1}{c}{Any}& \multicolumn{1}{c}{Private}& \multicolumn{1}{c}{2+ forms}& \multicolumn{1}{c}{Managed}  \\
       & \multicolumn{1}{c}{Medicare} & \multicolumn{1}{c}{insurance}& \multicolumn{1}{c}{ coverage}& \multicolumn{1}{c}{ of coverage}& \multicolumn{1}{c}{care}  \\
    \cmidrule[0.2pt](l){2-2}\cmidrule[0.2pt](l){3-3}\cmidrule[0.2pt](l){4-4}\cmidrule[0.2pt](l){5-5}\cmidrule[0.2pt](l){6-6}
& (1)& (2) & (3)& (4) & (5) \\   \hline
& &   \\
&\multicolumn{5}{c}{Panel A: RD Estimates at age 65 (2012)}   \\ \cmidrule[0.2pt](l){2-6}
Overall sample      &       0.622\sym{***}&       0.093\sym{***}&       0.170\sym{***}&       0.581\sym{***}&      -0.360\sym{***}\\
                    &     (0.067)         &     (0.011)         &     (0.023)         &     (0.055)         &     (0.049)         \\
Observations        &       11772         &       11769         &       11823         &       11823         &       10465         \\

& &   \\
\hline
& &   \\
&\multicolumn{5}{c}{Panel B: RD Estimates at age 65 (2014)}   \\ \cmidrule[0.2pt](l){2-6}
Overall sample      &       0.641\sym{***}&       0.049\sym{***}&       0.147\sym{***}&       0.586\sym{***}&      -0.311\sym{***}\\
                    &     (0.054)         &     (0.013)         &     (0.018)         &     (0.051)         &     (0.045)         \\
Observations        &       13377         &       13375         &       13442         &       13442         &       12364         \\

& &   \\
\hline
& &   \\
&\multicolumn{5}{c}{Panel C:  Diff-in-discs Estimates}   \\ \cmidrule[0.2pt](l){2-6}
Overall sample      &       0.019         &      -0.044\sym{***}&      -0.022         &       0.005         &       0.049         \\
                    &     (0.022)         &     (0.015)         &     (0.019)         &     (0.024)         &     (0.043)         \\
Observations        &       25149         &       25144         &       25265         &       25265         &       22829         \\

& &   \\
\hline \hline
\multicolumn{6}{p{11.5cm}}{\scriptsize{\textbf{Notes:} All columns in Panels A and B report RD estimates at age 65 using data from the Northeast, Midwest, and West regions in 2012 (Panel A) and 2014 (panel B). All columns in Panel C report the difference-in-discontinuities estimates using data from the Northeast, Midwest, and West regions, and compare outcomes in 2012 and 2014. The models include quadratic controls for age, fully interacted with dummies for age 65 or older and 2014. Other controls in these models include indicators for gender, race/ethnicity, education and region. Samples for the regression models only include people between the ages of 55 and 75.  Standard errors (in parentheses) are clustered by quarter of age.}} \\
\end{tabular}
}
\end{center}
\end{table}
% I would suggest adding commas to your figures for number of observations (e.g. "11,772" instead of "11772"). I think this would make them easier to read.

%%%%%%%%%%%%%%%%%%%%%%%%%%%%%%%%%%%%%%%%%%%%%%%%%

\begin{table}
\begin{center}
{
\renewcommand{\arraystretch}{0.6}
\setlength{\tabcolsep}{3pt}
\captionsetup{font={normalsize,bf}}
\caption {Access to Care}  \label{Tab:access1}
\footnotesize
\centering
\begin{tabular}{lcccccc}
\hline\hline
& &   \\
&\multicolumn{4}{c}{Panel A: Baseline measures}   \\ \cmidrule[0.2pt](l){2-5}
    & \multicolumn{1}{c}{Delayed care} & \multicolumn{1}{c}{Did not get}& \multicolumn{1}{c}{Saw doctor}& \multicolumn{1}{c}{Hospital stay} \\
       & \multicolumn{1}{c}{last year} & \multicolumn{1}{c}{care last year}& \multicolumn{1}{c}{last year}& \multicolumn{1}{c}{last year} \\
    \cmidrule[0.2pt](l){2-2}\cmidrule[0.2pt](l){3-3}\cmidrule[0.2pt](l){4-4}\cmidrule[0.2pt](l){5-5}
& (1)& (2)& (3)& (4)  \\   \hline
& &   \\
Overall sample      &       0.036\sym{*}  &       0.016         &       0.058         &       0.048\sym{*}  \\
                    &     (0.020)         &     (0.015)         &     (0.038)         &     (0.028)         \\
Observations        &       25530         &       25530         &       10462         &       25504         \\

& &   \\
\hline
& &   \\
&\multicolumn{4}{c}{Panel B: Alternative measures}   \\ \cmidrule[0.2pt](l){2-5}
        & \multicolumn{1}{c}{Could not} & \multicolumn{1}{c}{Could not}& \multicolumn{1}{c}{Could not}& \multicolumn{1}{c}{Could not get} \\
    & \multicolumn{1}{c}{afford prescription} & \multicolumn{1}{c}{afford to see a}& \multicolumn{1}{c}{afford follow-up}& \multicolumn{1}{c}{appointment soon} \\
       & \multicolumn{1}{c}{medicine last year} & \multicolumn{1}{c}{specialist last year}& \multicolumn{1}{c}{care last year}& \multicolumn{1}{c}{enough last year} \\
    \cmidrule[0.2pt](l){2-2}\cmidrule[0.2pt](l){3-3}\cmidrule[0.2pt](l){4-4}\cmidrule[0.2pt](l){5-5}
    & (1)& (2)& (3)& (4)  \\   \hline
& &   \\
Overall sample      &       0.070\sym{***}&       0.072\sym{***}&       0.055\sym{***}&      -0.026         \\
                    &     (0.025)         &     (0.024)         &     (0.015)         &     (0.029)         \\
Observations        &       12591         &       12588         &       12588         &       12596         \\

& &   \\
\hline \hline
\multicolumn{5}{p{14cm}}{\scriptsize{\textbf{Notes:} All columns report the fuzzy difference-in-discontinuities estimates using data from the Northeast, Midwest, and West regions, and compare outcomes in 2012 and 2014. The models include linear controls for age, fully interacted with dummies for age 65 or older and 2014. Other controls in these models include indicators for gender, race/ethnicity, education and region. Samples for the regression models only include people between the ages of 55 and 75.  Standard errors (in parentheses) are clustered by quarter of age.} } \\
\end{tabular}
}
\end{center}
\end{table}
% I would suggest adding commas to your figures for number of observations. I think this would make them easier to read.

%%%%%%%%%%%%%%%%%%%%%%%%%%%%%%%%%%%%%%%%%%%%%%%%%

\begin{table}
\begin{center}
{
\renewcommand{\arraystretch}{0.6}
\setlength{\tabcolsep}{3pt}
\captionsetup{font={normalsize,bf}}
\caption {Access to Care by Ethnicity}  \label{Tab:accesset}
\footnotesize
\centering
\begin{tabular}{lcccccc}
\hline\hline
& &   \\
&\multicolumn{4}{c}{Panel A: Baseline measures}   \\ \cmidrule[0.2pt](l){2-5}
    & \multicolumn{1}{c}{Delayed care} & \multicolumn{1}{c}{Did not get}& \multicolumn{1}{c}{Saw doctor}& \multicolumn{1}{c}{Hospital stay} \\
       & \multicolumn{1}{c}{last year} & \multicolumn{1}{c}{care last year}& \multicolumn{1}{c}{last year}& \multicolumn{1}{c}{last year} \\
    \cmidrule[0.2pt](l){2-2}\cmidrule[0.2pt](l){3-3}\cmidrule[0.2pt](l){4-4}\cmidrule[0.2pt](l){5-5}
& (1)& (2)& (3)& (4)  \\   \hline
& &   \\
White non-Hispanic (all)&       0.033         &       0.011         &       0.048         &       0.051\sym{*}  \\
                    &     (0.023)         &     (0.016)         &     (0.042)         &     (0.030)         \\
Observations        &       18596         &       18596         &        7730         &       18577         \\

& &   \\
Black non-Hispanic (all)&      -0.021         &      -0.069         &       0.367\sym{**} &      -0.035         \\
                    &     (0.100)         &     (0.097)         &     (0.165)         &     (0.122)         \\
Observations        &        1954         &        1954         &         884         &        1952         \\

& &   \\
Hispanic (all)      &       0.132         &       0.156\sym{**} &      -0.105         &       0.125         \\
                    &     (0.089)         &     (0.062)         &     (0.125)         &     (0.102)         \\
Observations        &        2939         &        2939         &        1111         &        2936         \\

& &   \\
Black or Hispanic (all)&       0.066         &       0.059         &       0.107         &       0.056         \\
                    &     (0.069)         &     (0.052)         &     (0.108)         &     (0.080)         \\
Observations        &        4893         &        4893         &        1995         &        4888         \\

& &   \\
Non-White (all)     &       0.055         &       0.046         &       0.100         &       0.038         \\
                    &     (0.052)         &     (0.044)         &     (0.101)         &     (0.063)         \\
Observations        &        6934         &        6934         &        2732         &        6927         \\

& &   \\
\hline
& &   \\
&\multicolumn{4}{c}{Panel B: Alternative measures}   \\ \cmidrule[0.2pt](l){2-5}
        & \multicolumn{1}{c}{Could not} & \multicolumn{1}{c}{Could not}& \multicolumn{1}{c}{Could not}& \multicolumn{1}{c}{Could not get} \\
    & \multicolumn{1}{c}{afford prescription} & \multicolumn{1}{c}{afford to see a}& \multicolumn{1}{c}{afford follow-up}& \multicolumn{1}{c}{appointment soon} \\
       & \multicolumn{1}{c}{medicine last year} & \multicolumn{1}{c}{specialist last year}& \multicolumn{1}{c}{care last year}& \multicolumn{1}{c}{enough last year} \\
    \cmidrule[0.2pt](l){2-2}\cmidrule[0.2pt](l){3-3}\cmidrule[0.2pt](l){4-4}\cmidrule[0.2pt](l){5-5}
    & (1)& (2)& (3)& (4)  \\   \hline
& &   \\
White non-Hispanic (all)&       0.067\sym{**} &       0.073\sym{***}&       0.035\sym{**} &      -0.026         \\
                    &     (0.027)         &     (0.027)         &     (0.017)         &     (0.029)         \\
Observations        &        9337         &        9335         &        9335         &        9343         \\

& &   \\
Black non-Hispanic (all)&      -0.047         &      -0.028         &       0.046         &      -0.072         \\
                    &     (0.190)         &     (0.074)         &     (0.056)         &     (0.096)         \\
Observations        &        1086         &        1086         &        1086         &        1085         \\

& &   \\
Hispanic (all)      &       0.233\sym{*}  &       0.086         &       0.153         &       0.040         \\
                    &     (0.124)         &     (0.103)         &     (0.102)         &     (0.083)         \\
Observations        &        1322         &        1322         &        1323         &        1323         \\

& &   \\
Black or Hispanic (all)&       0.108         &       0.031         &       0.105\sym{*}  &      -0.011         \\
                    &     (0.116)         &     (0.060)         &     (0.060)         &     (0.073)         \\
Observations        &        2408         &        2408         &        2409         &        2408         \\

& &   \\
Non-White (all)     &       0.083         &       0.068         &       0.151\sym{***}&      -0.017         \\
                    &     (0.088)         &     (0.045)         &     (0.049)         &     (0.058)         \\
Observations        &        3254         &        3253         &        3253         &        3253         \\

& &   \\
\hline \hline
\multicolumn{5}{p{15.3cm}}{\scriptsize{\textbf{Notes:} All columns report the fuzzy difference-in-discontinuities estimates using data from the Northeast, Midwest, and West regions, and compare outcomes in 2012 and 2014. The models include linear control sfor age, fully interacted with dummies for age 65 or older and 2014. Other controls in these models include indicators for gender, race/ethnicity, education and region. Samples for the regression models only include people between the ages of 55 and 75. Standard errors (in parentheses) are clustered by quarter of age.} } \\
\end{tabular}
}
\end{center}
\end{table}

\clearpage
% Your use of "non-Hispanic" in the table is strange to me. It implies that there are white Hispanic people as well as black Hispanic people, which I was not aware was possible.
% I would suggest adding commas to your figures for number of observations. I think this would make them easier to read.

%%%%%%%%%%%%%%%%%%%%%%%%%%%%%%%%%%%%%%%%%%%%%%%%%

\begin{table}
\begin{center}
{
\renewcommand{\arraystretch}{0.6}
\setlength{\tabcolsep}{3pt}
\captionsetup{font={normalsize,bf}}
\caption {Access to Care by Education}  \label{Tab:accessed}
\footnotesize
\centering
\begin{tabular}{lcccccc}
\hline\hline
& &   \\
&\multicolumn{4}{c}{Panel A: Baseline measures}   \\ \cmidrule[0.2pt](l){2-5}
    & \multicolumn{1}{c}{Delayed care} & \multicolumn{1}{c}{Did not get}& \multicolumn{1}{c}{Saw doctor}& \multicolumn{1}{c}{Hospital stay} \\
       & \multicolumn{1}{c}{last year} & \multicolumn{1}{c}{care last year}& \multicolumn{1}{c}{last year}& \multicolumn{1}{c}{last year} \\
    \cmidrule[0.2pt](l){2-2}\cmidrule[0.2pt](l){3-3}\cmidrule[0.2pt](l){4-4}\cmidrule[0.2pt](l){5-5}
& (1)& (2)& (3)& (4)  \\   \hline
& &   \\
High school dropout &       0.116         &       0.091         &      -0.057         &       0.064         \\
                    &     (0.094)         &     (0.082)         &     (0.193)         &     (0.111)         \\
Observations        &        3385         &        3385         &        1268         &        3382         \\

& &   \\
High school graduate&       0.029         &       0.025         &       0.071         &       0.063         \\
                    &     (0.058)         &     (0.041)         &     (0.090)         &     (0.051)         \\
Observations        &        7023         &        7023         &        2828         &        7016         \\

& &   \\
At least some college&       0.031         &       0.004         &       0.066         &       0.042         \\
                    &     (0.025)         &     (0.019)         &     (0.046)         &     (0.035)         \\
Observations        &       15122         &       15122         &        6366         &       15106         \\

& &   \\
\hline
& &   \\
&\multicolumn{4}{c}{Panel B: Alternative measures}   \\ \cmidrule[0.2pt](l){2-5}
        & \multicolumn{1}{c}{Could not} & \multicolumn{1}{c}{Could not}& \multicolumn{1}{c}{Could not}& \multicolumn{1}{c}{Could not get} \\
    & \multicolumn{1}{c}{afford prescription} & \multicolumn{1}{c}{afford to see a}& \multicolumn{1}{c}{afford follow-up}& \multicolumn{1}{c}{appointment soon} \\
       & \multicolumn{1}{c}{medicine last year} & \multicolumn{1}{c}{specialist last year}& \multicolumn{1}{c}{care last year}& \multicolumn{1}{c}{enough last year} \\
    \cmidrule[0.2pt](l){2-2}\cmidrule[0.2pt](l){3-3}\cmidrule[0.2pt](l){4-4}\cmidrule[0.2pt](l){5-5}
    & (1)& (2)& (3)& (4)  \\   \hline
& &   \\
High school dropout &       0.132         &       0.228\sym{*}  &       0.307\sym{***}&       0.105         \\
                    &     (0.128)         &     (0.131)         &     (0.107)         &     (0.100)         \\
Observations        &        1602         &        1601         &        1601         &        1601         \\

& &   \\
High school graduate&      -0.060         &       0.059         &       0.019         &      -0.065         \\
                    &     (0.063)         &     (0.046)         &     (0.050)         &     (0.059)         \\
Observations        &        3357         &        3356         &        3356         &        3359         \\

& &   \\
At least some college&       0.114\sym{***}&       0.061\sym{**} &       0.041\sym{**} &      -0.032         \\
                    &     (0.030)         &     (0.026)         &     (0.020)         &     (0.038)         \\
Observations        &        7632         &        7631         &        7631         &        7636         \\

& &   \\
\hline \hline
\multicolumn{5}{p{15.3cm}}{\scriptsize{\textbf{Notes:} All columns report the fuzzy difference-in-discontinuities estimates using data from the Northeast, Midwest, and West regions, and compare outcomes in 2012 and 2014. The models include linear controls for age, fully interacted with dummies for age 65 or older and 2014. Other controls in these models include indicators for gender, race/ethnicity, education and region. Samples for the regression models only include people between the ages of 55 and 75. Standard errors (in parentheses) are clustered by quarter of age.} } \\
\end{tabular}
}
\end{center}
\end{table}
% You may want to clarify in your notes whether "High school graduate" just means people who graduated high school but did not do any post-secondary education, or whether it includes everyone who graduated high school (including people who subsequently attended college, which would mean that individuals in your "At least some college" measure would also show up in your "High school graduate" measure).
% I would suggest adding commas to your figures for number of observations. I think this would make them easier to read.

%%%%%%%%%%%%%%%%%%%%%%%%%%%%%%%%%%%%%%%%%%%%%%%%%

\begin{table}
\begin{center}
{
\renewcommand{\arraystretch}{0.5}
\setlength{\tabcolsep}{3pt}
\captionsetup{font={normalsize,bf}}
\caption {Access to Care by Ethnicity and Education}  \label{Tab:accesseted}
\footnotesize
\centering
\begin{tabular}{lcccccc}
\hline\hline
& &   \\
&\multicolumn{4}{c}{Panel A: Baseline measures}   \\ \cmidrule[0.2pt](l){2-5}
    & \multicolumn{1}{c}{Delayed care} & \multicolumn{1}{c}{Did not get}& \multicolumn{1}{c}{Saw doctor}& \multicolumn{1}{c}{Hospital stay} \\
       & \multicolumn{1}{c}{last year} & \multicolumn{1}{c}{care last year}& \multicolumn{1}{c}{last year}& \multicolumn{1}{c}{last year} \\
    \cmidrule[0.2pt](l){2-2}\cmidrule[0.2pt](l){3-3}\cmidrule[0.2pt](l){4-4}\cmidrule[0.2pt](l){5-5}
& (1)& (2)& (3)& (4)  \\   \hline
& &   \\
White non-Hispanic: & &   \\   \cmidrule[0.2pt](l){1-1}
& &   \\
High school dropout &       0.062         &       0.061         &      -0.234         &       0.260         \\
                    &     (0.201)         &     (0.159)         &     (0.336)         &     (0.198)         \\
Observations        &        1366         &        1366         &         569         &        1365         \\

& &   \\
High school graduate&       0.061         &       0.039         &       0.089         &       0.079         \\
                    &     (0.062)         &     (0.046)         &     (0.095)         &     (0.054)         \\
Observations        &        5230         &        5230         &        2095         &        5224         \\

& &   \\
At least some college&       0.021         &      -0.004         &       0.051         &       0.031         \\
                    &     (0.027)         &     (0.019)         &     (0.048)         &     (0.035)         \\
Observations        &       12000         &       12000         &        5066         &       11988         \\

& &   \\
Minority: & &   \\    \cmidrule[0.2pt](l){1-1}
& &   \\
High school dropout &       0.170         &       0.107         &       0.116         &      -0.051         \\
                    &     (0.121)         &     (0.116)         &     (0.226)         &     (0.129)         \\
Observations        &        2019         &        2019         &         699         &        2017         \\

& &   \\
High school graduate&      -0.115         &      -0.037         &      -0.000         &      -0.024         \\
                    &     (0.113)         &     (0.070)         &     (0.149)         &     (0.111)         \\
Observations        &        1793         &        1793         &         733         &        1792         \\

& &   \\
At least some college&       0.094         &       0.054         &       0.141         &       0.120         \\
                    &     (0.062)         &     (0.057)         &     (0.134)         &     (0.101)         \\
Observations        &        3122         &        3122         &        1300         &        3118         \\

& &   \\
\hline
& &   \\
&\multicolumn{4}{c}{Panel B: Alternative measures}   \\ \cmidrule[0.2pt](l){2-5}
        & \multicolumn{1}{c}{Could not} & \multicolumn{1}{c}{Could not}& \multicolumn{1}{c}{Could not}& \multicolumn{1}{c}{Could not get} \\
    & \multicolumn{1}{c}{afford prescription} & \multicolumn{1}{c}{afford to see a}& \multicolumn{1}{c}{afford follow-up}& \multicolumn{1}{c}{appointment soon} \\
       & \multicolumn{1}{c}{medicine last year} & \multicolumn{1}{c}{specialist last year}& \multicolumn{1}{c}{care last year}& \multicolumn{1}{c}{enough last year} \\
    \cmidrule[0.2pt](l){2-2}\cmidrule[0.2pt](l){3-3}\cmidrule[0.2pt](l){4-4}\cmidrule[0.2pt](l){5-5}
    & (1)& (2)& (3)& (4)  \\   \hline
& &   \\
White non-Hispanic: & &   \\   \cmidrule[0.2pt](l){1-1}
& &   \\
High school dropout &       0.078         &       0.309         &       0.321\sym{*}  &       0.140         \\
                    &     (0.255)         &     (0.236)         &     (0.183)         &     (0.186)         \\
Observations        &         723         &         722         &         722         &         724         \\

& &   \\
High school graduate&      -0.077         &       0.054         &       0.010         &      -0.055         \\
                    &     (0.066)         &     (0.046)         &     (0.055)         &     (0.059)         \\
Observations        &        2495         &        2495         &        2495         &        2496         \\

& &   \\
At least some college&       0.116\sym{***}&       0.068\sym{**} &       0.028         &      -0.029         \\
                    &     (0.031)         &     (0.028)         &     (0.020)         &     (0.039)         \\
Observations        &        6119         &        6118         &        6118         &        6123         \\

& &   \\
Minority: & &   \\    \cmidrule[0.2pt](l){1-1}
& &   \\
High school dropout &       0.133         &       0.210\sym{*}  &       0.315\sym{**} &       0.094         \\
                    &     (0.222)         &     (0.115)         &     (0.125)         &     (0.086)         \\
Observations        &         879         &         879         &         879         &         877         \\

& &   \\
High school graduate&       0.000         &       0.074         &       0.053         &      -0.090         \\
                    &     (0.120)         &     (0.110)         &     (0.102)         &     (0.108)         \\
Observations        &         862         &         861         &         861         &         863         \\

& &   \\
At least some college&       0.105         &       0.005         &       0.130\sym{**} &      -0.031         \\
                    &     (0.109)         &     (0.065)         &     (0.060)         &     (0.086)         \\
Observations        &        1513         &        1513         &        1513         &        1513         \\

& &   \\
\hline \hline
\multicolumn{5}{p{15.3cm}}{\scriptsize{\textbf{Notes:} All columns report the fuzzy difference-in-discontinuities estimates using data from the Northeast, Midwest, and West regions, and compare outcomes in 2012 and 2014. The models include linear controls for age, fully interacted with dummies for age 65 or older and 2014. Other controls in these models include indicators for gender, race/ethnicity, education and region. Samples for the regression models only include people between the ages of 55 and 75. Standard errors (in parentheses) are clustered by quarter of age.} } \\
\end{tabular}
}
\end{center}
\end{table}
% Your use of "White non-Hispanic" is confusing. Are there white Hispanics?
% You may want to clarify in your notes whether "High school graduate" just means people who graduated high school but did not do any post-secondary education, or whether it includes everyone who graduated high school (including people who subsequently attended college, which would mean that individuals in your "At least some college" measure would also show up in your "High school graduate" measure).
% I would suggest adding commas to your figures for number of observations. I think this would make them easier to read.

%%%%%%%%%%%%%%%%%%%%%%%%%%%%%%%%%%%%%%%%%%%%%%%%%%

\begin{table}
\begin{center}
{
\renewcommand{\arraystretch}{0.7}
\setlength{\tabcolsep}{10pt}
\captionsetup{font={normalsize,bf}}
\caption {Employment}  \label{Tab:emp}
\footnotesize
\centering  \begin{tabular}{lcc}
\hline\hline
& &   \\
    & \multicolumn{1}{c}{Employed} & \multicolumn{1}{c}{Full time}  \\
    \cmidrule[0.2pt](l){2-2}\cmidrule[0.2pt](l){3-3}
& (1)& (2) \\   \hline
& &   \\
Overall sample      &      -0.020         &      -0.020         \\
                    &     (0.039)         &     (0.029)         \\
Observations        &       25159         &       25265         \\

& &   \\
\hline     \hline
\multicolumn{3}{p{7cm}}{\scriptsize{\textbf{Notes:} All columns report the difference-in-discontinuities estimates using data from the Northeast, Midwest, and West regions, and compare outcomes in 2012 and 2014. The models include quadratic controls for age, fully interacted with dummies for age 65 or older and 2014. Other controls in these models include indicators for gender, race/ethnicity, education and region. Samples for the regression models only include people between the ages of 55 and 75. Standard errors (in parentheses) are clustered by quarter of age.} } \\
\end{tabular}
}
\end{center}
\end{table}
% I would suggest adding commas to your figures for number of observations. I think this would make them easier to read.

%%%%%%%%%%%%%%%%%%%%%%%%%%%%%%%%%%%%%%%%%%%%%%%%%%
\begin{table}
\begin{center}
{
\renewcommand{\arraystretch}{0.6}
\setlength{\tabcolsep}{4pt}
\captionsetup{font={normalsize,bf}}
\caption {Insurance Coverage by Ethnicity}  \label{}
\footnotesize
\centering  \begin{tabular}{lccccc}
\hline\hline
& &   \\
    & \multicolumn{1}{c}{On} & \multicolumn{1}{c}{Any}& \multicolumn{1}{c}{Private}& \multicolumn{1}{c}{2+ forms}& \multicolumn{1}{c}{Managed}  \\
       & \multicolumn{1}{c}{Medicare} & \multicolumn{1}{c}{insurance}& \multicolumn{1}{c}{ coverage}& \multicolumn{1}{c}{ of coverage}& \multicolumn{1}{c}{care}  \\
    \cmidrule[0.2pt](l){2-2}\cmidrule[0.2pt](l){3-3}\cmidrule[0.2pt](l){4-4}\cmidrule[0.2pt](l){5-5}\cmidrule[0.2pt](l){6-6}
& (1)& (2) & (3)& (4) & (5) \\   \hline
& &   \\
White non-Hispanic (all)&       0.036         &      -0.054\sym{***}&      -0.021         &       0.021         &       0.009         \\
                    &     (0.025)         &     (0.019)         &     (0.022)         &     (0.024)         &     (0.045)         \\
Observations        &       18319         &       18315         &       18387         &       18387         &       16938         \\

& &   \\
Black non-Hispanic (all)&       0.014         &      -0.007         &       0.054         &      -0.035         &       0.231         \\
                    &     (0.100)         &     (0.062)         &     (0.092)         &     (0.100)         &     (0.149)         \\
Observations        &        1920         &        1920         &        1935         &        1935         &        1703         \\

& &   \\
Hispanic (all)      &      -0.038         &      -0.062         &      -0.211\sym{***}&      -0.040         &       0.127         \\
                    &     (0.104)         &     (0.049)         &     (0.066)         &     (0.101)         &     (0.114)         \\
Observations        &        2891         &        2890         &        2915         &        2915         &        2389         \\

& &   \\
Black or Hispanic (all)&      -0.018         &      -0.042         &      -0.119\sym{**} &      -0.033         &       0.163\sym{*}  \\
                    &     (0.069)         &     (0.037)         &     (0.047)         &     (0.072)         &     (0.089)         \\
Observations        &        4811         &        4810         &        4850         &        4850         &        4092         \\

& &   \\
Non-White (all)     &      -0.045         &      -0.023         &      -0.024         &      -0.056         &       0.201\sym{***}\\
                    &     (0.054)         &     (0.025)         &     (0.034)         &     (0.058)         &     (0.073)         \\
Observations        &        6830         &        6829         &        6878         &        6878         &        5891         \\

& &   \\
\hline\hline
\multicolumn{6}{p{12cm}}{\scriptsize{\textbf{Notes:} All columns report the difference-in-discontinuities estimates using data from the Northeast, Midwest, and West regions, and compare outcomes in 2012 and 2014. The models include quadratic controls for age, fully interacted with dummies for age 65 or older and 2014. Other controls in these models include indicators for gender, race/ethnicity, education and region. Samples for the regression models only include people between the ages of 55 and 75. Standard errors (in parentheses) are clustered by quarter of age.} } \\
\end{tabular}
}
\end{center}
\end{table}
% I would suggest adding commas to your figures for number of observations. I think this would make them easier to read.
% Your use of "non-Hispanic" in the table is strange to me. It implies that there are white Hispanic people as well as black Hispanic people, which I was not aware was possible.

%%%%%%%%%%%%%%%%%%%%%%%%%%%%%%%%%%%%%%%%%%%%%%%%%%
\begin{table}
\begin{center}
{
\renewcommand{\arraystretch}{0.6}
\setlength{\tabcolsep}{4pt}
\captionsetup{font={normalsize,bf}}
\caption {Insurance Coverage by Education}  \label{}
\footnotesize
\centering  \begin{tabular}{lccccc}
\hline\hline
& &   \\
    & \multicolumn{1}{c}{On} & \multicolumn{1}{c}{Any}& \multicolumn{1}{c}{Private}& \multicolumn{1}{c}{2+ forms}& \multicolumn{1}{c}{Managed}  \\
       & \multicolumn{1}{c}{Medicare} & \multicolumn{1}{c}{insurance}& \multicolumn{1}{c}{ coverage}& \multicolumn{1}{c}{ of coverage}& \multicolumn{1}{c}{care}  \\
    \cmidrule[0.2pt](l){2-2}\cmidrule[0.2pt](l){3-3}\cmidrule[0.2pt](l){4-4}\cmidrule[0.2pt](l){5-5}\cmidrule[0.2pt](l){6-6}
& (1)& (2) & (3)& (4) & (5) \\   \hline
& &   \\
High school dropout &      -0.029         &      -0.044         &       0.025         &      -0.029         &       0.152         \\
                    &     (0.073)         &     (0.071)         &     (0.079)         &     (0.076)         &     (0.098)         \\
Observations        &        3325         &        3324         &        3347         &        3347         &        2797         \\

& &   \\
High school graduate&      -0.025         &      -0.062\sym{*}  &      -0.045         &      -0.073         &       0.086         \\
                    &     (0.053)         &     (0.036)         &     (0.041)         &     (0.058)         &     (0.085)         \\
Observations        &        6900         &        6898         &        6933         &        6933         &        6221         \\

& &   \\
At least some college&       0.051         &      -0.038\sym{**} &      -0.028         &       0.047         &       0.021         \\
                    &     (0.032)         &     (0.017)         &     (0.031)         &     (0.031)         &     (0.047)         \\
Observations        &       14924         &       14922         &       14985         &       14985         &       13811         \\

& &   \\
\hline\hline
\multicolumn{6}{p{12cm}}{\scriptsize{\textbf{Notes:} All columns report the difference-in-discontinuities estimates using data from the Northeast, Midwest, and West regions, and compare outcomes in 2012 and 2014. The models include quadratic controls for age, fully interacted with dummies for age 65 or older and 2014. Other controls in these models include indicators for gender, race/ethnicity, education and region. Samples for the regression models only include people between the ages of 55 and 75. Standard errors (in parentheses) are clustered by quarter of age.} } \\
\end{tabular}
}
\end{center}
\end{table}
% I would suggest adding commas to your figures for number of observations. I think this would make them easier to read.
% You may want to clarify in your notes whether "High school graduate" just means people who graduated high school but did not do any post-secondary education, or whether it includes everyone who graduated high school (including people who subsequently attended college, which would mean that individuals in your "At least some college" measure would also show up in your "High school graduate" measure).

%%%%%%%%%%%%%
%%%table
\begin{table}
\begin{center}
{
\renewcommand{\arraystretch}{0.6}
\setlength{\tabcolsep}{4pt}
\captionsetup{font={normalsize,bf}}
\caption {Employment  by Ethnicity}  \label{Tab:empet}
\footnotesize
\centering  \begin{tabular}{lcc}
\hline\hline
& &   \\
    & \multicolumn{1}{c}{Employed} & \multicolumn{1}{c}{Full time}  \\
    \cmidrule[0.2pt](l){2-2}\cmidrule[0.2pt](l){3-3}
& (1)& (2) \\   \hline
& &   \\
White non-Hispanic (all)&      -0.032         &      -0.035         \\
                    &     (0.047)         &     (0.038)         \\
Observations        &       18330         &       18387         \\

& &   \\
Black non-Hispanic (all)&      -0.148         &      -0.068         \\
                    &     (0.101)         &     (0.100)         \\
Observations        &        1924         &        1935         \\

& &   \\
Hispanic (all)      &       0.017         &       0.027         \\
                    &     (0.075)         &     (0.068)         \\
Observations        &        2905         &        2915         \\

& &   \\
Black or Hispanic (all)&      -0.050         &      -0.018         \\
                    &     (0.063)         &     (0.045)         \\
Observations        &        4829         &        4850         \\

& &   \\
Non-White (all)     &       0.020         &       0.029         \\
                    &     (0.054)         &     (0.042)         \\
Observations        &        6829         &        6878         \\

& &   \\
\hline\hline
\multicolumn{3}{p{7cm}}{\scriptsize{\textbf{Notes:} All columns report the difference-in-discontinuities estimates using data from the Northeast, Midwest, and West regions, and compare outcomes in 2012 and 2014. The models include quadratic controls for age, fully interacted with dummies for age 65 or older and 2014. Other controls in these models include indicators for gender, race/ethnicity, education and region. Samples for the regression models only include people between the ages of 55 and 75. Standard errors (in parentheses) are clustered by quarter of age.} } \\
\end{tabular}
}
\end{center}
\end{table}
% I would suggest adding commas to your figures for number of observations. I think this would make them easier to read.
% Your use of "non-Hispanic" in the table is strange to me. It implies that there are white Hispanic people as well as black Hispanic people, which I was not aware was possible.

%%%%%%%%%%%%%
%%%table
\begin{table}
\begin{center}
{
\renewcommand{\arraystretch}{0.6}
\setlength{\tabcolsep}{4pt}
\captionsetup{font={normalsize,bf}}
\caption {Employment  by Education}  \label{Tab:emped}
\footnotesize
\centering  \begin{tabular}{lcc}
\hline\hline
& &   \\
    & \multicolumn{1}{c}{Employed} & \multicolumn{1}{c}{Full time}  \\
    \cmidrule[0.2pt](l){2-2}\cmidrule[0.2pt](l){3-3}
& (1)& (2) \\   \hline
& &   \\
High school dropout &       0.006         &      -0.000         \\
                    &     (0.077)         &     (0.058)         \\
Observations        &        3336         &        3347         \\

& &   \\
High school graduate&      -0.022         &      -0.084         \\
                    &     (0.073)         &     (0.071)         \\
Observations        &        6904         &        6933         \\

& &   \\
At least some college&      -0.031         &      -0.004         \\
                    &     (0.050)         &     (0.039)         \\
Observations        &       14919         &       14985         \\

& &   \\
\hline\hline
\multicolumn{3}{p{8cm}}{\scriptsize{\textbf{Notes:} All columns report the difference-in-discontinuities estimates using data from the Northeast, Midwest, and West regions, and compare outcomes in 2012 and 2014. The models include quadratic controls for age, fully interacted with dummies for age 65 or older and 2014. Other controls in these models include indicators for gender, race/ethnicity, education and region. Samples for the regression models only include people between the ages of 55 and 75. Standard errors (in parentheses) are clustered by quarter of age.} } \\
\end{tabular}
}
\end{center}
\end{table}
% I would suggest adding commas to your figures for number of observations. I think this would make them easier to read.
% You may want to clarify in your notes whether "High school graduate" just means people who graduated high school but did not do any post-secondary education, or whether it includes everyone who graduated high school (including people who subsequently attended college, which would mean that individuals in your "At least some college" measure would also show up in your "High school graduate" measure).

%%%%%%%%%%%%%%%%%%%%%%%%%%%%%%%%%%%%%%%%%%%%%%%%%

\begin{table}
\begin{center}
{
\renewcommand{\arraystretch}{0.6}
\setlength{\tabcolsep}{3pt}
\captionsetup{font={normalsize,bf}}
\caption {Access to Care (with Quadratic Controls for Age)}  \label{tab:agesquared1}
\footnotesize
\centering
\begin{tabular}{lcccccc}
\hline\hline
& &   \\
&\multicolumn{4}{c}{Panel A: Baseline measures}   \\ \cmidrule[0.2pt](l){2-5}
    & \multicolumn{1}{c}{Delayed care} & \multicolumn{1}{c}{Did not get}& \multicolumn{1}{c}{Saw doctor}& \multicolumn{1}{c}{Hospital stay} \\
       & \multicolumn{1}{c}{last year} & \multicolumn{1}{c}{care last year}& \multicolumn{1}{c}{last year}& \multicolumn{1}{c}{last year} \\
    \cmidrule[0.2pt](l){2-2}\cmidrule[0.2pt](l){3-3}\cmidrule[0.2pt](l){4-4}\cmidrule[0.2pt](l){5-5}
& (1)& (2)& (3)& (4)  \\   \hline
& &   \\
Overall sample      &       0.059\sym{*}  &       0.036         &       0.010         &       0.043         \\
                    &     (0.032)         &     (0.024)         &     (0.059)         &     (0.037)         \\
Observations        &       25530         &       25530         &       10462         &       25504         \\

& &   \\
\hline
& &   \\
&\multicolumn{4}{c}{Panel B: Alternative measures}   \\ \cmidrule[0.2pt](l){2-5}
        & \multicolumn{1}{c}{Could not} & \multicolumn{1}{c}{Could not}& \multicolumn{1}{c}{Could not}& \multicolumn{1}{c}{Could not get} \\
    & \multicolumn{1}{c}{afford prescription} & \multicolumn{1}{c}{afford to see a}& \multicolumn{1}{c}{afford follow-up}& \multicolumn{1}{c}{appointment soon} \\
       & \multicolumn{1}{c}{medicine last year} & \multicolumn{1}{c}{specialist last year}& \multicolumn{1}{c}{care last year}& \multicolumn{1}{c}{enough last year} \\
    \cmidrule[0.2pt](l){2-2}\cmidrule[0.2pt](l){3-3}\cmidrule[0.2pt](l){4-4}\cmidrule[0.2pt](l){5-5}
    & (1)& (2)& (3)& (4)  \\   \hline
& &   \\
Overall sample      &       0.116\sym{***}&       0.091\sym{**} &       0.040\sym{*}  &      -0.078         \\
                    &     (0.038)         &     (0.039)         &     (0.023)         &     (0.047)         \\
Observations        &       12591         &       12588         &       12588         &       12596         \\

& &   \\
\hline \hline
\multicolumn{5}{p{14cm}}{\scriptsize{\textbf{Notes:} All columns report the fuzzy difference-in-discontinuities estimates using data from the Northeast, Midwest, and West regions, and compare outcomes in 2012 and 2014. The models include quadratic controls for age, fully interacted with dummies for age 65 or older and 2014. Other controls in these models include indicators for gender, race/ethnicity, education and region. Samples for the regression models only include people between the ages of 55 and 75. Standard errors (in parentheses) are clustered by quarter of age.} } \\
\end{tabular}
}
\end{center}
\end{table}
% I would suggest adding commas to your figures for number of observations. I think this would make them easier to read.

%%%%%%%%%%%%%%%%%%%%%%%%%%%%%%%%%%%%%%%%%%%%%%%%%

\begin{table}
\begin{center}
{
\renewcommand{\arraystretch}{0.6}
\setlength{\tabcolsep}{3pt}
\captionsetup{font={normalsize,bf}}
\caption {Access to Care by Ethnicity (with Quadratic Controls for Age)}  \label{tab:agesquared2}
\footnotesize
\centering
\begin{tabular}{lcccccc}
\hline\hline
& &   \\
&\multicolumn{4}{c}{Panel A: Baseline measures}   \\ \cmidrule[0.2pt](l){2-5}
    & \multicolumn{1}{c}{Delayed care} & \multicolumn{1}{c}{Did not get}& \multicolumn{1}{c}{Saw doctor}& \multicolumn{1}{c}{Hospital stay} \\
       & \multicolumn{1}{c}{last year} & \multicolumn{1}{c}{care last year}& \multicolumn{1}{c}{last year}& \multicolumn{1}{c}{last year} \\
    \cmidrule[0.2pt](l){2-2}\cmidrule[0.2pt](l){3-3}\cmidrule[0.2pt](l){4-4}\cmidrule[0.2pt](l){5-5}
& (1)& (2)& (3)& (4)  \\   \hline
& &   \\
White non-Hispanic (all)&       0.051         &       0.039         &       0.010         &       0.067         \\
                    &     (0.034)         &     (0.024)         &     (0.067)         &     (0.042)         \\
Observations        &       18596         &       18596         &        7730         &       18577         \\

& &   \\
Black non-Hispanic (all)&       0.154         &      -0.212         &       0.435         &      -0.121         \\
                    &     (0.132)         &     (0.175)         &     (0.273)         &     (0.192)         \\
Observations        &        1954         &        1954         &         884         &        1952         \\

& &   \\
Hispanic (all)      &       0.136         &       0.192\sym{*}  &      -0.492\sym{**} &       0.036         \\
                    &     (0.159)         &     (0.099)         &     (0.220)         &     (0.194)         \\
Observations        &        2939         &        2939         &        1111         &        2936         \\

& &   \\
Black or Hispanic (all)&       0.140         &       0.014         &      -0.025         &      -0.035         \\
                    &     (0.113)         &     (0.095)         &     (0.194)         &     (0.138)         \\
Observations        &        4893         &        4893         &        1995         &        4888         \\

& &   \\
Non-White (all)     &       0.119         &       0.038         &       0.002         &      -0.073         \\
                    &     (0.085)         &     (0.080)         &     (0.173)         &     (0.109)         \\
Observations        &        6934         &        6934         &        2732         &        6927         \\

& &   \\
\hline
& &   \\
&\multicolumn{4}{c}{Panel B: Alternative measures}   \\ \cmidrule[0.2pt](l){2-5}
        & \multicolumn{1}{c}{Could not} & \multicolumn{1}{c}{Could not}& \multicolumn{1}{c}{Could not}& \multicolumn{1}{c}{Could not get} \\
    & \multicolumn{1}{c}{afford prescription} & \multicolumn{1}{c}{afford to see a}& \multicolumn{1}{c}{afford follow-up}& \multicolumn{1}{c}{appointment soon} \\
       & \multicolumn{1}{c}{medicine last year} & \multicolumn{1}{c}{specialist last year}& \multicolumn{1}{c}{care last year}& \multicolumn{1}{c}{enough last year} \\
    \cmidrule[0.2pt](l){2-2}\cmidrule[0.2pt](l){3-3}\cmidrule[0.2pt](l){4-4}\cmidrule[0.2pt](l){5-5}
    & (1)& (2)& (3)& (4)  \\   \hline
& &   \\
White non-Hispanic (all)&       0.123\sym{***}&       0.102\sym{**} &       0.044         &      -0.094\sym{**} \\
                    &     (0.036)         &     (0.044)         &     (0.027)         &     (0.045)         \\
Observations        &        9337         &        9335         &        9335         &        9343         \\

& &   \\
Black non-Hispanic (all)&      -0.036         &      -0.010         &      -0.026         &       0.011         \\
                    &     (0.364)         &     (0.134)         &     (0.087)         &     (0.197)         \\
Observations        &        1086         &        1086         &        1086         &        1085         \\

& &   \\
Hispanic (all)      &       0.244         &       0.027         &       0.034         &       0.119         \\
                    &     (0.241)         &     (0.185)         &     (0.194)         &     (0.187)         \\
Observations        &        1322         &        1322         &        1323         &        1323         \\

& &   \\
Black or Hispanic (all)&       0.122         &       0.010         &       0.005         &       0.090         \\
                    &     (0.221)         &     (0.103)         &     (0.105)         &     (0.155)         \\
Observations        &        2408         &        2408         &        2409         &        2408         \\

& &   \\
Non-White (all)     &       0.093         &       0.042         &       0.024         &       0.016         \\
                    &     (0.164)         &     (0.068)         &     (0.067)         &     (0.110)         \\
Observations        &        3254         &        3253         &        3253         &        3253         \\

& &   \\
\hline \hline
\multicolumn{5}{p{15.5cm}}{\scriptsize{\textbf{Notes:} All columns report the fuzzy difference-in-discontinuities estimates using data from the Northeast, Midwest, and West regions, and compare outcomes in 2012 and 2014. The models include quadratic controls for age, fully interacted with dummies for age 65 or older and 2014. Other controls in these models include indicators for gender, race/ethnicity, education and region. Samples for the regression models only include people between the ages of 55 and 75. Standard errors (in parentheses) are clustered by quarter of age.} } \\
\end{tabular}
}
\end{center}
\end{table}
% I would suggest adding commas to your figures for number of observations. I think this would make them easier to read.
% Your use of "non-Hispanic" in the table is strange to me. It implies that there are white Hispanic people as well as black Hispanic people, which I was not aware was possible.

%%%%%%%%%%%%%%%%%%%%%%%%%%%%%%%%%%%%%%%%%%%%%%%%%

\begin{table}
\begin{center}
{
\renewcommand{\arraystretch}{0.6}
\setlength{\tabcolsep}{3pt}
\captionsetup{font={normalsize,bf}}
\caption {Access to Care by Education (with Quadratic Controls for Age)}  \label{tab:agesquared3}
\footnotesize
\centering
\begin{tabular}{lcccccc}
\hline\hline
& &   \\
&\multicolumn{4}{c}{Panel A: Baseline measures}   \\ \cmidrule[0.2pt](l){2-5}
    & \multicolumn{1}{c}{Delayed care} & \multicolumn{1}{c}{Did not get}& \multicolumn{1}{c}{Saw doctor}& \multicolumn{1}{c}{Hospital stay} \\
       & \multicolumn{1}{c}{last year} & \multicolumn{1}{c}{care last year}& \multicolumn{1}{c}{last year}& \multicolumn{1}{c}{last year} \\
    \cmidrule[0.2pt](l){2-2}\cmidrule[0.2pt](l){3-3}\cmidrule[0.2pt](l){4-4}\cmidrule[0.2pt](l){5-5}
& (1)& (2)& (3)& (4)  \\   \hline
& &   \\
High school dropout &       0.253         &       0.131         &      -0.790         &      -0.015         \\
                    &     (0.179)         &     (0.184)         &     (0.536)         &     (0.251)         \\
Observations        &        3385         &        3385         &        1268         &        3382         \\

& &   \\
High school graduate&       0.009         &       0.083         &      -0.046         &       0.091\sym{*}  \\
                    &     (0.089)         &     (0.070)         &     (0.123)         &     (0.051)         \\
Observations        &        7023         &        7023         &        2828         &        7016         \\

& &   \\
At least some college&       0.075\sym{*}  &       0.009         &       0.111\sym{*}  &       0.023         \\
                    &     (0.038)         &     (0.027)         &     (0.059)         &     (0.048)         \\
Observations        &       15122         &       15122         &        6366         &       15106         \\

& &   \\
\hline
& &   \\
&\multicolumn{4}{c}{Panel B: Alternative measures}   \\ \cmidrule[0.2pt](l){2-5}
        & \multicolumn{1}{c}{Could not} & \multicolumn{1}{c}{Could not}& \multicolumn{1}{c}{Could not}& \multicolumn{1}{c}{Could not get} \\
    & \multicolumn{1}{c}{afford prescription} & \multicolumn{1}{c}{afford to see a}& \multicolumn{1}{c}{afford follow-up}& \multicolumn{1}{c}{appointment soon} \\
       & \multicolumn{1}{c}{medicine last year} & \multicolumn{1}{c}{specialist last year}& \multicolumn{1}{c}{care last year}& \multicolumn{1}{c}{enough last year} \\
    \cmidrule[0.2pt](l){2-2}\cmidrule[0.2pt](l){3-3}\cmidrule[0.2pt](l){4-4}\cmidrule[0.2pt](l){5-5}
    & (1)& (2)& (3)& (4)  \\   \hline
& &   \\
High school dropout &       0.073         &       0.244         &       0.233         &       0.139         \\
                    &     (0.301)         &     (0.242)         &     (0.174)         &     (0.230)         \\
Observations        &        1602         &        1601         &        1601         &        1601         \\

& &   \\
High school graduate&       0.029         &       0.114         &       0.033         &      -0.036         \\
                    &     (0.102)         &     (0.084)         &     (0.081)         &     (0.108)         \\
Observations        &        3357         &        3356         &        3356         &        3359         \\

& &   \\
At least some college&       0.148\sym{***}&       0.065         &       0.018         &      -0.118\sym{**} \\
                    &     (0.042)         &     (0.040)         &     (0.029)         &     (0.056)         \\
Observations        &        7632         &        7631         &        7631         &        7636         \\

& &   \\
\hline \hline
\multicolumn{5}{p{14.5cm}}{\scriptsize{\textbf{Notes:} All columns report the fuzzy difference-in-discontinuities estimates using data from the Northeast, Midwest, and West regions, and compare outcomes in 2012 and 2014. The models include quadratic controls for age, fully interacted with dummies for age 65 or older and 2014. Other controls in these models include indicators for gender, race/ethnicity, education and region. Samples for the regression models only include people between the ages of 55 and 75. Standard errors (in parentheses) are clustered by quarter of age.} } \\
\end{tabular}
}
\end{center}
\end{table}
% I would suggest adding commas to your figures for number of observations. I think this would make them easier to read.
% You may want to clarify in your notes whether "High school graduate" just means people who graduated high school but did not do any post-secondary education, or whether it includes everyone who graduated high school (including people who subsequently attended college, which would mean that individuals in your "At least some college" measure would also show up in your "High school graduate" measure).

%%%%%%%%%%%%%%%%%%%%%%%%%%%%%%%%%%%%%%%%%%%%%%%%%%%%%%%%%%%%%%

\begin{table}
\begin{center}
{
\renewcommand{\arraystretch}{0.6}
\setlength{\tabcolsep}{4pt}
\captionsetup{font={normalsize,bf}}
\caption {Insurance Coverage (Smaller Bandwidth)}  \label{}
\footnotesize
\centering  \begin{tabular}{lccccc}
\hline\hline
    & \multicolumn{1}{c}{On} & \multicolumn{1}{c}{Any}& \multicolumn{1}{c}{Private}& \multicolumn{1}{c}{2+ forms}& \multicolumn{1}{c}{Managed}  \\
       & \multicolumn{1}{c}{Medicare} & \multicolumn{1}{c}{insurance}& \multicolumn{1}{c}{ coverage}& \multicolumn{1}{c}{ of coverage}& \multicolumn{1}{c}{care}  \\
    \cmidrule[0.2pt](l){2-2}\cmidrule[0.2pt](l){3-3}\cmidrule[0.2pt](l){4-4}\cmidrule[0.2pt](l){5-5}\cmidrule[0.2pt](l){6-6}
& (1)& (2) & (3)& (4) & (5) \\   \hline
& &   \\
Overall sample      &       0.017         &      -0.046\sym{***}&      -0.035\sym{*}  &       0.000         &       0.052         \\
                    &     (0.021)         &     (0.015)         &     (0.019)         &     (0.023)         &     (0.042)         \\
Observations        &       12677         &       12674         &       12729         &       12729         &       11570         \\

& &   \\
\hline
& &   \\
Classified by ethnicity: & &   \\    \cmidrule[0.2pt](l){1-1}
& &   \\
White non-Hispanic (all)&       0.036         &      -0.039\sym{**} &      -0.025         &       0.019         &       0.013         \\
                    &     (0.024)         &     (0.015)         &     (0.020)         &     (0.024)         &     (0.043)         \\
Observations        &        9376         &        9373         &        9411         &        9411         &        8690         \\

& &   \\
Black non-Hispanic (all)&       0.015         &      -0.053         &       0.021         &      -0.054         &       0.122         \\
                    &     (0.096)         &     (0.056)         &     (0.092)         &     (0.096)         &     (0.143)         \\
Observations        &         904         &         904         &         908         &         908         &         813         \\

& &   \\
Hispanic (all)      &      -0.069         &      -0.106\sym{**} &      -0.205\sym{***}&      -0.070         &       0.231\sym{**} \\
                    &     (0.100)         &     (0.050)         &     (0.063)         &     (0.097)         &     (0.111)         \\
Observations        &        1393         &        1393         &        1400         &        1400         &        1172         \\

& &   \\
Black or Hispanic (all)&      -0.037         &      -0.090\sym{**} &      -0.126\sym{**} &      -0.060         &       0.182\sym{**} \\
                    &     (0.067)         &     (0.034)         &     (0.047)         &     (0.066)         &     (0.085)         \\
Observations        &        2297         &        2297         &        2308         &        2308         &        1985         \\

& &   \\
Non-White (all)     &      -0.060         &      -0.065\sym{**} &      -0.045         &      -0.070         &       0.212\sym{***}\\
                    &     (0.052)         &     (0.025)         &     (0.037)         &     (0.055)         &     (0.071)         \\
Observations        &        3301         &        3301         &        3318         &        3318         &        2880         \\

& &   \\
Classified by education: & &   \\    \cmidrule[0.2pt](l){1-1}
& &   \\
High school dropout &      -0.017         &      -0.078         &       0.038         &      -0.023         &       0.172\sym{*}  \\
                    &     (0.071)         &     (0.073)         &     (0.082)         &     (0.074)         &     (0.096)         \\
Observations        &        1569         &        1569         &        1576         &        1576         &        1340         \\

& &   \\
High school graduate&      -0.041         &      -0.068\sym{**} &      -0.085\sym{**} &      -0.090         &       0.078         \\
                    &     (0.052)         &     (0.032)         &     (0.041)         &     (0.055)         &     (0.082)         \\
Observations        &        3288         &        3286         &        3306         &        3306         &        2967         \\

& &   \\
At least some college&       0.050\sym{*}  &      -0.030\sym{*}  &      -0.028         &       0.045         &       0.029         \\
                    &     (0.030)         &     (0.016)         &     (0.029)         &     (0.031)         &     (0.046)         \\
Observations        &        7820         &        7819         &        7847         &        7847         &        7263         \\

& &   \\
\hline         \hline
\multicolumn{6}{p{13cm}}{\scriptsize{\textbf{Notes:} All columns report the difference-in-discontinuities estimates using data from the Northeast, Midwest, and West regions, and compare outcomes in 2012 and 2014. The models include linear controls for age, fully interacted with dummies for age 65 or older and 2014.  Other controls in these models include indicators for gender, race/ethnicity, education and region. Samples for regression models only include people between the ages of 60 and 70.  Standard errors (in parentheses) are clustered by quarter of age.} }\\
\end{tabular}
}
\end{center}
\end{table}
% I would suggest adding commas to your figures for number of observations. I think this would make them easier to read.
% You may want to clarify in your notes whether "High school graduate" just means people who graduated high school but did not do any post-secondary education, or whether it includes everyone who graduated high school (including people who subsequently attended college, which would mean that individuals in your "At least some college" measure would also show up in your "High school graduate" measure).
% Your use of "non-Hispanic" in the table is strange to me. It implies that there are white Hispanic people as well as black Hispanic people, which I was not aware was possible.

%%%%%%%%%%%%%
%%%table
\begin{table}
\begin{center}
{
\renewcommand{\arraystretch}{0.6}
\setlength{\tabcolsep}{4pt}
\captionsetup{font={normalsize,bf}}
\caption {Employment (Smaller Bandwidth)}  \label{Tab:empsmall}
\footnotesize
\centering  \begin{tabular}{lcc}
\hline\hline
& &   \\
    & \multicolumn{1}{c}{Employed} & \multicolumn{1}{c}{Full time}  \\
    \cmidrule[0.2pt](l){2-2}\cmidrule[0.2pt](l){3-3}
& (1)& (2) \\   \hline
& &   \\
Overall sample      &      -0.019         &      -0.020         \\
                    &     (0.037)         &     (0.028)         \\
Observations        &       12687         &       12729         \\

& &   \\
\hline
& &   \\
Classified by ethnicity: & &   \\    \cmidrule[0.2pt](l){1-1}
White non-Hispanic (all)&      -0.029         &      -0.037         \\
                    &     (0.046)         &     (0.035)         \\
Observations        &        9387         &        9411         \\

& &   \\
Black non-Hispanic (all)&      -0.130         &      -0.085         \\
                    &     (0.103)         &     (0.099)         \\
Observations        &         903         &         908         \\

& &   \\
Hispanic (all)      &       0.010         &       0.061         \\
                    &     (0.079)         &     (0.069)         \\
Observations        &        1397         &        1400         \\

& &   \\
Black or Hispanic (all)&      -0.046         &      -0.001         \\
                    &     (0.062)         &     (0.044)         \\
Observations        &        2300         &        2308         \\

& &   \\
Non-White (all)     &       0.013         &       0.035         \\
                    &     (0.054)         &     (0.041)         \\
Observations        &        3300         &        3318         \\

& &   \\
Classified by education: & &   \\    \cmidrule[0.2pt](l){1-1}
High school dropout &       0.003         &       0.004         \\
                    &     (0.072)         &     (0.053)         \\
Observations        &        1573         &        1576         \\

& &   \\
High school graduate&      -0.023         &      -0.083         \\
                    &     (0.070)         &     (0.067)         \\
Observations        &        3287         &        3306         \\

& &   \\
At least some college&      -0.026         &      -0.002         \\
                    &     (0.049)         &     (0.037)         \\
Observations        &        7827         &        7847         \\

& &   \\
\hline\hline
\multicolumn{3}{p{8cm}}{\scriptsize{\textbf{Notes:} All columns report the difference-in-discontinuities estimates using data from the Northeast, Midwest, and West regions, and compare outcomes in 2012 and 2014. The models include linear controls for age, fully interacted with dummies for age 65 or older and 2014.  Other controls in these models include indicators for gender, race/ethnicity, education and region. Samples for regression models only include people between the ages of 60 and 70.  Standard errors (in parentheses) are clustered by quarter of age.} }\\
\end{tabular}
}
\end{center}
\end{table}
% I would suggest adding commas to your figures for number of observations. I think this would make them easier to read.
% You may want to clarify in your notes whether "High school graduate" just means people who graduated high school but did not do any post-secondary education, or whether it includes everyone who graduated high school (including people who subsequently attended college, which would mean that individuals in your "At least some college" measure would also show up in your "High school graduate" measure).
% Your use of "non-Hispanic" in the table is strange to me. It implies that there are white Hispanic people as well as black Hispanic people, which I was not aware was possible.

%%%%%%%%%%%%%%%%%%%%%%%%%%%%%%%%%%%%%%%%%%%%%%%%%%%%%%%%%%%%%%

\begin{table}
\begin{center}
{
\renewcommand{\arraystretch}{0.6}
\setlength{\tabcolsep}{3pt}
\captionsetup{font={normalsize,bf}}
\caption {Access to Care (Smaller Bandwidth)}  \label{tab:smallerband1}
\footnotesize
\centering
\begin{tabular}{lcccccc}
\hline\hline
& &   \\
&\multicolumn{4}{c}{Panel A: Baseline measures}   \\ \cmidrule[0.2pt](l){2-5}
    & \multicolumn{1}{c}{Delayed care} & \multicolumn{1}{c}{Did not get}& \multicolumn{1}{c}{Saw doctor}& \multicolumn{1}{c}{Hospital stay} \\
       & \multicolumn{1}{c}{last year} & \multicolumn{1}{c}{care last year}& \multicolumn{1}{c}{last year}& \multicolumn{1}{c}{last year} \\
    \cmidrule[0.2pt](l){2-2}\cmidrule[0.2pt](l){3-3}\cmidrule[0.2pt](l){4-4}\cmidrule[0.2pt](l){5-5}
& (1)& (2)& (3)& (4)  \\   \hline
& &   \\
Overall sample      &       0.050\sym{*}  &       0.027         &       0.031         &       0.015         \\
                    &     (0.029)         &     (0.021)         &     (0.052)         &     (0.034)         \\
Observations        &       13294         &       13294         &        5607         &       13288         \\

& &   \\
\hline
& &   \\
&\multicolumn{4}{c}{Panel B: Alternative measures}   \\ \cmidrule[0.2pt](l){2-5}
        & \multicolumn{1}{c}{Could not} & \multicolumn{1}{c}{Could not}& \multicolumn{1}{c}{Could not}& \multicolumn{1}{c}{Could not get} \\
    & \multicolumn{1}{c}{afford prescription} & \multicolumn{1}{c}{afford to see a}& \multicolumn{1}{c}{afford follow-up}& \multicolumn{1}{c}{appointment soon} \\
       & \multicolumn{1}{c}{medicine last year} & \multicolumn{1}{c}{specialist last year}& \multicolumn{1}{c}{care last year}& \multicolumn{1}{c}{enough last year} \\
    \cmidrule[0.2pt](l){2-2}\cmidrule[0.2pt](l){3-3}\cmidrule[0.2pt](l){4-4}\cmidrule[0.2pt](l){5-5}
    & (1)& (2)& (3)& (4)  \\   \hline
& &   \\
Overall sample      &       0.083\sym{**} &       0.083\sym{**} &       0.038\sym{*}  &      -0.055         \\
                    &     (0.031)         &     (0.035)         &     (0.021)         &     (0.043)         \\
Observations        &        6724         &        6723         &        6722         &        6724         \\

& &   \\
\hline \hline
\multicolumn{5}{p{14cm}}{\scriptsize{\textbf{Notes:} All columns report the fuzzy difference-in-discontinuities estimates using data from the Northeast, Midwest, and West regions, and compare outcomes in 2012 and 2014. The models include linear controls for age, fully interacted with dummies for age 65 or older and 2014. Other controls in these models include indicators for gender, race/ethnicity, education and region. Samples for the regression models only include people between the ages of 55 and 75. Standard errors (in parentheses) are clustered by quarter of age.} } \\
\end{tabular}
}
\end{center}
\end{table}
% Your heading says "Smaller Bandwith", but in your notes, you say you're using 55 to 75 (instead of 60 to 70). Is there a mistake somewhere?
% I would suggest adding commas to your figures for number of observations. I think this would make them easier to read.

%%%%%%%%%%%%%%%%%%%%%%%%%%%%%%%%%%%%%%%%%%%%%%%%%%%%%%%%%%%%%%

\begin{table}
\begin{center}
{
\renewcommand{\arraystretch}{0.6}
\setlength{\tabcolsep}{3pt}
\captionsetup{font={normalsize,bf}}
\caption {Access to Care by Ethnicity (Smaller Bandwidth)}  \label{tab:smallerband2}
\footnotesize
\centering
\begin{tabular}{lcccccc}
\hline\hline
& &   \\
&\multicolumn{4}{c}{Panel A: Baseline measures}   \\ \cmidrule[0.2pt](l){2-5}
    & \multicolumn{1}{c}{Delayed care} & \multicolumn{1}{c}{Did not get}& \multicolumn{1}{c}{Saw doctor}& \multicolumn{1}{c}{Hospital stay} \\
       & \multicolumn{1}{c}{last year} & \multicolumn{1}{c}{care last year}& \multicolumn{1}{c}{last year}& \multicolumn{1}{c}{last year} \\
    \cmidrule[0.2pt](l){2-2}\cmidrule[0.2pt](l){3-3}\cmidrule[0.2pt](l){4-4}\cmidrule[0.2pt](l){5-5}
& (1)& (2)& (3)& (4)  \\   \hline
& &   \\
White non-Hispanic (all)&       0.045         &       0.026         &       0.028         &       0.025         \\
                    &     (0.032)         &     (0.023)         &     (0.059)         &     (0.038)         \\
Observations        &        9801         &        9801         &        4176         &        9797         \\

& &   \\
Black non-Hispanic (all)&       0.143         &      -0.120         &       0.380         &      -0.123         \\
                    &     (0.130)         &     (0.168)         &     (0.300)         &     (0.171)         \\
Observations        &         965         &         965         &         438         &         964         \\

& &   \\
Hispanic (all)      &       0.117         &       0.171\sym{*}  &      -0.272         &       0.096         \\
                    &     (0.141)         &     (0.089)         &     (0.177)         &     (0.173)         \\
Observations        &        1468         &        1468         &         594         &        1468         \\

& &   \\
Black or Hispanic (all)&       0.124         &       0.046         &       0.021         &       0.009         \\
                    &     (0.104)         &     (0.091)         &     (0.180)         &     (0.119)         \\
Observations        &        2433         &        2433         &        1032         &        2432         \\

& &   \\
Non-White (all)     &       0.086         &       0.044         &       0.039         &      -0.027         \\
                    &     (0.077)         &     (0.073)         &     (0.160)         &     (0.097)         \\
Observations        &        3493         &        3493         &        1431         &        3491         \\

& &   \\
\hline
& &   \\
&\multicolumn{4}{c}{Panel B: Alternative measures}   \\ \cmidrule[0.2pt](l){2-5}
        & \multicolumn{1}{c}{Could not} & \multicolumn{1}{c}{Could not}& \multicolumn{1}{c}{Could not}& \multicolumn{1}{c}{Could not get} \\
    & \multicolumn{1}{c}{afford prescription} & \multicolumn{1}{c}{afford to see a}& \multicolumn{1}{c}{afford follow-up}& \multicolumn{1}{c}{appointment soon} \\
       & \multicolumn{1}{c}{medicine last year} & \multicolumn{1}{c}{specialist last year}& \multicolumn{1}{c}{care last year}& \multicolumn{1}{c}{enough last year} \\
    \cmidrule[0.2pt](l){2-2}\cmidrule[0.2pt](l){3-3}\cmidrule[0.2pt](l){4-4}\cmidrule[0.2pt](l){5-5}
    & (1)& (2)& (3)& (4)  \\   \hline
& &   \\
White non-Hispanic (all)&       0.069\sym{**} &       0.087\sym{**} &       0.031         &      -0.067         \\
                    &     (0.030)         &     (0.042)         &     (0.026)         &     (0.042)         \\
Observations        &        5023         &        5022         &        5021         &        5025         \\

& &   \\
Black non-Hispanic (all)&       0.086         &       0.013         &       0.021         &      -0.058         \\
                    &     (0.337)         &     (0.111)         &     (0.079)         &     (0.166)         \\
Observations        &         539         &         539         &         539         &         538         \\

& &   \\
Hispanic (all)      &       0.322         &       0.090         &       0.102         &       0.129         \\
                    &     (0.210)         &     (0.163)         &     (0.162)         &     (0.137)         \\
Observations        &         707         &         707         &         707         &         707         \\

& &   \\
Black or Hispanic (all)&       0.209         &       0.052         &       0.054         &       0.059         \\
                    &     (0.199)         &     (0.090)         &     (0.090)         &     (0.129)         \\
Observations        &        1246         &        1246         &        1246         &        1245         \\

& &   \\
Non-White (all)     &       0.167         &       0.080         &       0.083         &       0.016         \\
                    &     (0.149)         &     (0.061)         &     (0.059)         &     (0.092)         \\
Observations        &        1701         &        1701         &        1701         &        1699         \\

& &   \\
\hline \hline
\multicolumn{5}{p{15.5cm}}{\scriptsize{\textbf{Notes:} All columns report the fuzzy difference-in-discontinuities estimates using data from the Northeast, Midwest, and West regions, and compare outcomes in 2012 and 2014. The models include linear controls for age, fully interacted with dummies for age 65 or older and 2014. Other controls in these models include indicators for gender, race/ethnicity, education and region. Samples for the regression models only include people between the ages of 55 and 75. Standard errors (in parentheses) are clustered by quarter of age.} } \\
\end{tabular}
}
\end{center}
\end{table}
% Your heading says "Smaller Bandwith", but in your notes, you say you're using 55 to 75 (instead of 60 to 70). Is there a mistake somewhere?
% I would suggest adding commas to your figures for number of observations. I think this would make them easier to read.
% Your use of "non-Hispanic" in the table is strange to me. It implies that there are white Hispanic people as well as black Hispanic people, which I was not aware was possible.

%%%%%%%%%%%%%%%%%%%%%%%%%%%%%%%%%%%%%%%%%%%%%%%%%%%%%%%%%%%%%%

\begin{table}
\begin{center}
{
\renewcommand{\arraystretch}{0.6}
\setlength{\tabcolsep}{3pt}
\captionsetup{font={normalsize,bf}}
\caption {Access to Care by Education (Smaller Bandwidth)}  \label{tab:smallerband3}
\footnotesize
\centering
\begin{tabular}{lcccccc}
\hline\hline
& &   \\
&\multicolumn{4}{c}{Panel A: Baseline measures}   \\ \cmidrule[0.2pt](l){2-5}
    & \multicolumn{1}{c}{Delayed care} & \multicolumn{1}{c}{Did not get}& \multicolumn{1}{c}{Saw doctor}& \multicolumn{1}{c}{Hospital stay} \\
       & \multicolumn{1}{c}{last year} & \multicolumn{1}{c}{care last year}& \multicolumn{1}{c}{last year}& \multicolumn{1}{c}{last year} \\
    \cmidrule[0.2pt](l){2-2}\cmidrule[0.2pt](l){3-3}\cmidrule[0.2pt](l){4-4}\cmidrule[0.2pt](l){5-5}
& (1)& (2)& (3)& (4)  \\   \hline
& &   \\
High school dropout &       0.177         &       0.109         &      -0.285         &      -0.065         \\
                    &     (0.162)         &     (0.142)         &     (0.343)         &     (0.203)         \\
Observations        &        1651         &        1651         &         636         &        1651         \\

& &   \\
High school graduate&       0.000         &       0.066         &       0.063         &       0.097\sym{*}  \\
                    &     (0.086)         &     (0.064)         &     (0.121)         &     (0.057)         \\
Observations        &        3454         &        3454         &        1448         &        3450         \\

& &   \\
At least some college&       0.061\sym{*}  &       0.003         &       0.049         &      -0.011         \\
                    &     (0.034)         &     (0.025)         &     (0.057)         &     (0.044)         \\
Observations        &        8189         &        8189         &        3523         &        8187         \\

& &   \\
\hline
& &   \\
&\multicolumn{4}{c}{Panel B: Alternative measures}   \\ \cmidrule[0.2pt](l){2-5}
        & \multicolumn{1}{c}{Could not} & \multicolumn{1}{c}{Could not}& \multicolumn{1}{c}{Could not}& \multicolumn{1}{c}{Could not get} \\
    & \multicolumn{1}{c}{afford prescription} & \multicolumn{1}{c}{afford to see a}& \multicolumn{1}{c}{afford follow-up}& \multicolumn{1}{c}{appointment soon} \\
       & \multicolumn{1}{c}{medicine last year} & \multicolumn{1}{c}{specialist last year}& \multicolumn{1}{c}{care last year}& \multicolumn{1}{c}{enough last year} \\
    \cmidrule[0.2pt](l){2-2}\cmidrule[0.2pt](l){3-3}\cmidrule[0.2pt](l){4-4}\cmidrule[0.2pt](l){5-5}
    & (1)& (2)& (3)& (4)  \\   \hline
& &   \\
High school dropout &      -0.026         &       0.271         &       0.324\sym{**} &       0.166         \\
                    &     (0.225)         &     (0.165)         &     (0.138)         &     (0.181)         \\
Observations        &         802         &         801         &         801         &         802         \\

& &   \\
High school graduate&       0.007         &       0.106         &       0.044         &      -0.026         \\
                    &     (0.092)         &     (0.074)         &     (0.080)         &     (0.105)         \\
Observations        &        1697         &        1697         &        1697         &        1697         \\

& &   \\
At least some college&       0.123\sym{***}&       0.059         &       0.007         &      -0.086         \\
                    &     (0.039)         &     (0.040)         &     (0.028)         &     (0.052)         \\
Observations        &        4225         &        4225         &        4224         &        4225         \\

& &   \\
\hline \hline
\multicolumn{5}{p{14.5cm}}{\scriptsize{\textbf{Notes:} All columns report the fuzzy difference-in-discontinuities estimates using data from the Northeast, Midwest, and West regions, and compare outcomes in 2012 and 2014. The models include linear controls for age, fully interacted with dummies for age 65 or older and 2014. Other controls in these models include indicators for gender, race/ethnicity, education and region. Samples for the regression models only include people between the ages of 55 and 75. Standard errors (in parentheses) are clustered by quarter of age.} } \\
\end{tabular}
}
\end{center}
\end{table}
% Your heading says "Smaller Bandwith", but in your notes, you say you're using 55 to 75 (instead of 60 to 70). Is there a mistake somewhere?
% I would suggest adding commas to your figures for number of observations. I think this would make them easier to read.
% Your use of "non-Hispanic" in the table is strange to me. It implies that there are white Hispanic people as well as black Hispanic people, which I was not aware was possible.

%%%%%%%%%%%%%%%%%%%%%%%%%%%%%%%%%%%%%%%%%%%%%%%%%
%%%%%%%%%%%%%%%%%%%%%%%%%%%%%%%%%%%%%%%%%%%%%%%%%

\begin{table}
\begin{center}
{
\renewcommand{\arraystretch}{0.6}
\setlength{\tabcolsep}{2pt}
\captionsetup{font={normalsize,bf}}
\caption {Insurance Coverage (Part D)}  \label{}
\footnotesize
\centering  \begin{tabular}{lccccc}
\hline\hline
& &   \\
    & \multicolumn{1}{c}{On Medicare Part D}   \\
    \cmidrule[0.2pt](l){2-2}
& (1) \\   \hline
&    \\
Overall sample      &       0.038         \\
                    &     (0.029)         \\
Observations        &       25149         \\

&    \\
\hline \hline
\multicolumn{2}{p{8.5cm}}{\scriptsize{\textbf{Notes:} All columns report the difference-in-discontinuities estimates using data from the Northeast, Midwest, and West regions, and compare outcomes in 2012 and 2014. The models include quadratic controls for age, fully interacted with dummies for age 65 or older and 2014. Other controls in these models include indicators for gender, race/ethnicity, education and region. Samples for the regression models only include people between the ages of 55 and 75. Standard errors (in parentheses) are clustered by quarter of age.}} \\
\end{tabular}
}
\end{center}
\end{table}
% I would suggest adding commas to your figures for number of observations. I think this would make them easier to read.

%%%%%%%%%%%%%%%%%%%%%%%%%%%%%%%%%%%%%%%%%%%%%%%%%

\begin{table}
\begin{center}
{
\renewcommand{\arraystretch}{0.6}
\setlength{\tabcolsep}{3pt}
\captionsetup{font={normalsize,bf}}
\caption {Access to Care (Part D)}  \label{tab:partD}
\footnotesize
\centering
\begin{tabular}{lcccccc}
\hline\hline
& &   \\
&\multicolumn{4}{c}{Panel A: Baseline measures}   \\ \cmidrule[0.2pt](l){2-5}
    & \multicolumn{1}{c}{Delayed care} & \multicolumn{1}{c}{Did not get}& \multicolumn{1}{c}{Saw doctor}& \multicolumn{1}{c}{Hospital stay} \\
       & \multicolumn{1}{c}{last year} & \multicolumn{1}{c}{care last year}& \multicolumn{1}{c}{last year}& \multicolumn{1}{c}{last year} \\
    \cmidrule[0.2pt](l){2-2}\cmidrule[0.2pt](l){3-3}\cmidrule[0.2pt](l){4-4}\cmidrule[0.2pt](l){5-5}
& (1)& (2)& (3)& (4)  \\   \hline
& &   \\
Overall sample      &       0.096\sym{*}  &       0.042         &       0.165         &       0.143\sym{*}  \\
                    &     (0.057)         &     (0.040)         &     (0.103)         &     (0.074)         \\
Observations        &       25530         &       25530         &       10462         &       25504         \\

& &   \\
\hline
& &   \\
&\multicolumn{4}{c}{Panel B: Alternative measures}   \\ \cmidrule[0.2pt](l){2-5}
        & \multicolumn{1}{c}{Could not} & \multicolumn{1}{c}{Could not}& \multicolumn{1}{c}{Could not}& \multicolumn{1}{c}{Could not get} \\
    & \multicolumn{1}{c}{afford prescription} & \multicolumn{1}{c}{afford to see a}& \multicolumn{1}{c}{afford follow-up}& \multicolumn{1}{c}{appointment soon} \\
       & \multicolumn{1}{c}{medicine last year} & \multicolumn{1}{c}{specialist last year}& \multicolumn{1}{c}{care last year}& \multicolumn{1}{c}{enough last year} \\
    \cmidrule[0.2pt](l){2-2}\cmidrule[0.2pt](l){3-3}\cmidrule[0.2pt](l){4-4}\cmidrule[0.2pt](l){5-5}
    & (1)& (2)& (3)& (4)  \\   \hline
& &   \\
Overall sample      &       0.188\sym{***}&       0.205\sym{***}&       0.161\sym{***}&      -0.065         \\
                    &     (0.069)         &     (0.066)         &     (0.042)         &     (0.082)         \\
Observations        &       12591         &       12588         &       12588         &       12596         \\

& &   \\
\hline \hline
\multicolumn{5}{p{14cm}}{\scriptsize{\textbf{Notes:} All columns report the fuzzy difference-in-discontinuities estimates using data from the Northeast, Midwest, and West regions, and compare outcomes in 2012 and 2014. The models include linear controls for age, fully interacted with dummies for age 65 or older and 2014. Other controls in these models include indicators for gender, race/ethnicity, education and region. Samples for the regression models only include people between the ages of 55 and 75. Standard errors (in parentheses) are clustered by quarter of age.} } \\
\end{tabular}
}
\end{center}
\end{table}
% I would suggest adding commas to your figures for number of observations. I think this would make them easier to read.

%%%%%%%%%%%%%%%%%%%%%%%%%%%%%%%%%%%%%%%%%%%%%%%%%

\begin{table}
\begin{center}
{
\renewcommand{\arraystretch}{0.6}
\setlength{\tabcolsep}{3pt}
\captionsetup{font={normalsize,bf}}
\caption {Access to Care by Ethnicity (Part D)}  \label{}
\footnotesize
\centering
\begin{tabular}{lcccccc}
\hline\hline
& &   \\
&\multicolumn{4}{c}{Panel A: Baseline measures}   \\ \cmidrule[0.2pt](l){2-5}
    & \multicolumn{1}{c}{Delayed care} & \multicolumn{1}{c}{Did not get}& \multicolumn{1}{c}{Saw doctor}& \multicolumn{1}{c}{Hospital stay} \\
       & \multicolumn{1}{c}{last year} & \multicolumn{1}{c}{care last year}& \multicolumn{1}{c}{last year}& \multicolumn{1}{c}{last year} \\
    \cmidrule[0.2pt](l){2-2}\cmidrule[0.2pt](l){3-3}\cmidrule[0.2pt](l){4-4}\cmidrule[0.2pt](l){5-5}
& (1)& (2)& (3)& (4)  \\   \hline
& &   \\
White non-Hispanic (all)&       0.097         &       0.031         &       0.148         &       0.157\sym{*}  \\
                    &     (0.066)         &     (0.047)         &     (0.110)         &     (0.084)         \\
Observations        &       18596         &       18596         &        7730         &       18577         \\

& &   \\
Black non-Hispanic (all)&       0.167         &      -0.168         &       0.958         &      -0.080         \\
                    &     (0.400)         &     (0.318)         &     (0.763)         &     (0.407)         \\
Observations        &        1954         &        1954         &         884         &        1952         \\

& &   \\
Hispanic (all)      &       0.244         &       0.337\sym{**} &      -0.210         &       0.302         \\
                    &     (0.218)         &     (0.153)         &     (0.434)         &     (0.236)         \\
Observations        &        2939         &        2939         &        1111         &        2936         \\

& &   \\
Black or Hispanic (all)&       0.162         &       0.144         &       0.239         &       0.159         \\
                    &     (0.184)         &     (0.137)         &     (0.363)         &     (0.206)         \\
Observations        &        4893         &        4893         &        1995         &        4888         \\

& &   \\
Non-White (all)     &       0.117         &       0.101         &       0.287         &       0.101         \\
                    &     (0.128)         &     (0.107)         &     (0.327)         &     (0.151)         \\
Observations        &        6934         &        6934         &        2732         &        6927         \\

& &   \\
\hline
& &   \\
&\multicolumn{4}{c}{Panel B: Alternative measures}   \\ \cmidrule[0.2pt](l){2-5}
        & \multicolumn{1}{c}{Could not} & \multicolumn{1}{c}{Could not}& \multicolumn{1}{c}{Could not}& \multicolumn{1}{c}{Could not get} \\
    & \multicolumn{1}{c}{afford prescription} & \multicolumn{1}{c}{afford to see a}& \multicolumn{1}{c}{afford follow-up}& \multicolumn{1}{c}{appointment soon} \\
       & \multicolumn{1}{c}{medicine last year} & \multicolumn{1}{c}{specialist last year}& \multicolumn{1}{c}{care last year}& \multicolumn{1}{c}{enough last year} \\
    \cmidrule[0.2pt](l){2-2}\cmidrule[0.2pt](l){3-3}\cmidrule[0.2pt](l){4-4}\cmidrule[0.2pt](l){5-5}
    & (1)& (2)& (3)& (4)  \\   \hline
& &   \\
White non-Hispanic (all)&       0.181\sym{**} &       0.211\sym{***}&       0.107\sym{**} &      -0.067         \\
                    &     (0.075)         &     (0.077)         &     (0.046)         &     (0.081)         \\
Observations        &        9337         &        9335         &        9335         &        9343         \\

& &   \\
Black non-Hispanic (all)&      -0.116         &       0.056         &       0.327         &      -0.361         \\
                    &     (0.885)         &     (0.340)         &     (0.327)         &     (0.522)         \\
Observations        &        1086         &        1086         &        1086         &        1085         \\

& &   \\
Hispanic (all)      &       0.566\sym{*}  &       0.298         &       0.439\sym{*}  &       0.090         \\
                    &     (0.338)         &     (0.267)         &     (0.257)         &     (0.207)         \\
Observations        &        1322         &        1322         &        1323         &        1323         \\

& &   \\
Black or Hispanic (all)&       0.291         &       0.087         &       0.315\sym{*}  &      -0.029         \\
                    &     (0.337)         &     (0.163)         &     (0.164)         &     (0.201)         \\
Observations        &        2408         &        2408         &        2409         &        2408         \\

& &   \\
Non-White (all)     &       0.208         &       0.183         &       0.428\sym{***}&      -0.020         \\
                    &     (0.242)         &     (0.123)         &     (0.145)         &     (0.158)         \\
Observations        &        3254         &        3253         &        3253         &        3253         \\

& &   \\
\hline \hline
\multicolumn{5}{p{15.3cm}}{\scriptsize{\textbf{Notes:} All columns report the fuzzy difference-in-discontinuities estimates using data from the Northeast, Midwest, and West regions, and compare outcomes in 2012 and 2014. The models include linear controls for age, fully interacted with dummies for age 65 or older and 2014. Other controls in these models include indicators for gender, race/ethnicity, education and region. Samples for the regression models only include people between the ages of 55 and 75. Standard errors (in parentheses) are clustered by quarter of age.} } \\
\end{tabular}
}
\end{center}
\end{table}

\clearpage
% I would suggest adding commas to your figures for number of observations. I think this would make them easier to read.
% Your use of "non-Hispanic" in the table is strange to me. It implies that there are white Hispanic people as well as black Hispanic people, which I was not aware was possible.

%%%%%%%%%%%%%%%%%%%%%%%%%%%%%%%%%%%%%%%%%%%%%%%%%

%%%%%%%%%%%%%%%%%%%%%%%%%%%%%%%%%%%%%%%%%%%%%%%%%

\begin{table}
\begin{center}
{
\renewcommand{\arraystretch}{0.6}
\setlength{\tabcolsep}{2pt}
\captionsetup{font={normalsize,bf}}
\caption {Insurance Coverage (Parts A, B or C)}  \label{}
\footnotesize
\centering  \begin{tabular}{lccccc}
\hline\hline
& &   \\
    & \multicolumn{1}{c}{On Medicare Part D}   \\
    \cmidrule[0.2pt](l){2-2}
& (1) \\   \hline
&    \\
Overall sample      &       0.012         \\
                    &     (0.027)         \\
Observations        &       20070         \\

&    \\
\hline \hline
\multicolumn{2}{p{8.5cm}}{\scriptsize{\textbf{Notes:} All columns report the difference-in-discontinuities estimates using data from the Northeast, Midwest, and West regions, and compare outcomes in 2012 and 2014. The models include quadratic controls for age, fully interacted with dummies for age 65 or older and 2014. Other controls in these models include indicators for gender, race/ethnicity, education and region. Samples for the regression models only include people between the ages of 55 and 75. Standard errors (in parentheses) are clustered by quarter of age.}} \\
\end{tabular}
}
\end{center}
\end{table}
% Your heading refers to Parts A, B and C, but then the column heading says "On Medicare Part D". Is this a mistake?
% I would suggest adding commas to your figures for number of observations. I think this would make them easier to read.

%%%%%%%%%%%%%%%%%%%%%%%%%%%%%%%%%%%%%%%%%%%%%%%%%

\begin{table}
\begin{center}
{
\renewcommand{\arraystretch}{0.6}
\setlength{\tabcolsep}{3pt}
\captionsetup{font={normalsize,bf}}
\caption {Access to Care (Parts A, B or C)}  \label{tab:partABC}
\footnotesize
\centering
\begin{tabular}{lcccccc}
\hline\hline
& &   \\
&\multicolumn{4}{c}{Panel A: Baseline measures}   \\ \cmidrule[0.2pt](l){2-5}
    & \multicolumn{1}{c}{Delayed care} & \multicolumn{1}{c}{Did not get}& \multicolumn{1}{c}{Saw doctor}& \multicolumn{1}{c}{Hospital stay} \\
       & \multicolumn{1}{c}{last year} & \multicolumn{1}{c}{care last year}& \multicolumn{1}{c}{last year}& \multicolumn{1}{c}{last year} \\
    \cmidrule[0.2pt](l){2-2}\cmidrule[0.2pt](l){3-3}\cmidrule[0.2pt](l){4-4}\cmidrule[0.2pt](l){5-5}
& (1)& (2)& (3)& (4)  \\   \hline
& &   \\
Overall sample      &       0.028         &      -0.005         &       0.074         &       0.042         \\
                    &     (0.024)         &     (0.017)         &     (0.048)         &     (0.030)         \\
Observations        &       20346         &       20346         &        8296         &       20328         \\

& &   \\
\hline
& &   \\
&\multicolumn{4}{c}{Panel B: Alternative measures}   \\ \cmidrule[0.2pt](l){2-5}
        & \multicolumn{1}{c}{Could not} & \multicolumn{1}{c}{Could not}& \multicolumn{1}{c}{Could not}& \multicolumn{1}{c}{Could not get} \\
    & \multicolumn{1}{c}{afford prescription} & \multicolumn{1}{c}{afford to see a}& \multicolumn{1}{c}{afford follow-up}& \multicolumn{1}{c}{appointment soon} \\
       & \multicolumn{1}{c}{medicine last year} & \multicolumn{1}{c}{specialist last year}& \multicolumn{1}{c}{care last year}& \multicolumn{1}{c}{enough last year} \\
    \cmidrule[0.2pt](l){2-2}\cmidrule[0.2pt](l){3-3}\cmidrule[0.2pt](l){4-4}\cmidrule[0.2pt](l){5-5}
    & (1)& (2)& (3)& (4)  \\   \hline
& &   \\
Overall sample      &       0.090\sym{***}&       0.064\sym{**} &       0.041\sym{**} &      -0.060         \\
                    &     (0.028)         &     (0.026)         &     (0.018)         &     (0.037)         \\
Observations        &        9767         &        9765         &        9765         &        9772         \\

& &   \\
\hline \hline
\multicolumn{5}{p{14cm}}{\scriptsize{\textbf{Notes:} All columns report the fuzzy difference-in-discontinuities estimates using data from the Northeast, Midwest, and West regions, and compare outcomes in 2012 and 2014. The models include linear controls for age, fully interacted with dummies for age 65 or older and 2014. Other controls in these models include indicators for gender, race/ethnicity, education and region. Samples for the regression models only include people between the ages of 55 and 75. Standard errors (in parentheses) are clustered by quarter of age.} } \\
\end{tabular}
}
\end{center}
\end{table}
% I would suggest adding commas to your figures for number of observations. I think this would make them easier to read.

%%%%%%%%%%%%%%%%%%%%%%%%%%%%%%%%%%%%%%%%%%%%%%%%%

\begin{table}
\begin{center}
{
\renewcommand{\arraystretch}{0.6}
\setlength{\tabcolsep}{3pt}
\captionsetup{font={normalsize,bf}}
\caption {Access to Care by Ethnicity (Parts A, B or C)}  \label{}
\footnotesize
\centering
\begin{tabular}{lcccccc}
\hline\hline
& &   \\
&\multicolumn{4}{c}{Panel A: Baseline measures}   \\ \cmidrule[0.2pt](l){2-5}
    & \multicolumn{1}{c}{Delayed care} & \multicolumn{1}{c}{Did not get}& \multicolumn{1}{c}{Saw doctor}& \multicolumn{1}{c}{Hospital stay} \\
       & \multicolumn{1}{c}{last year} & \multicolumn{1}{c}{care last year}& \multicolumn{1}{c}{last year}& \multicolumn{1}{c}{last year} \\
    \cmidrule[0.2pt](l){2-2}\cmidrule[0.2pt](l){3-3}\cmidrule[0.2pt](l){4-4}\cmidrule[0.2pt](l){5-5}
& (1)& (2)& (3)& (4)  \\   \hline
& &   \\
White non-Hispanic (all)&       0.030         &      -0.007         &       0.040         &       0.049         \\
                    &     (0.026)         &     (0.017)         &     (0.056)         &     (0.034)         \\
Observations        &       14759         &       14759         &        6123         &       14745         \\

& &   \\
Black non-Hispanic (all)&      -0.020         &      -0.089         &       0.501\sym{***}&      -0.005         \\
                    &     (0.129)         &     (0.147)         &     (0.182)         &     (0.143)         \\
Observations        &        1553         &        1553         &         686         &        1552         \\

& &   \\
Hispanic (all)      &       0.068         &       0.106         &       0.024         &       0.034         \\
                    &     (0.130)         &     (0.083)         &     (0.153)         &     (0.119)         \\
Observations        &        2396         &        2396         &         896         &        2395         \\

& &   \\
Black or Hispanic (all)&       0.034         &       0.019         &       0.241\sym{**} &       0.006         \\
                    &     (0.075)         &     (0.078)         &     (0.113)         &     (0.082)         \\
Observations        &        3949         &        3949         &        1582         &        3947         \\

& &   \\
Non-White (all)     &       0.026         &       0.010         &       0.214\sym{*}  &       0.014         \\
                    &     (0.067)         &     (0.068)         &     (0.112)         &     (0.062)         \\
Observations        &        5587         &        5587         &        2173         &        5583         \\

& &   \\
\hline
& &   \\
&\multicolumn{4}{c}{Panel B: Alternative measures}   \\ \cmidrule[0.2pt](l){2-5}
        & \multicolumn{1}{c}{Could not} & \multicolumn{1}{c}{Could not}& \multicolumn{1}{c}{Could not}& \multicolumn{1}{c}{Could not get} \\
    & \multicolumn{1}{c}{afford prescription} & \multicolumn{1}{c}{afford to see a}& \multicolumn{1}{c}{afford follow-up}& \multicolumn{1}{c}{appointment soon} \\
       & \multicolumn{1}{c}{medicine last year} & \multicolumn{1}{c}{specialist last year}& \multicolumn{1}{c}{care last year}& \multicolumn{1}{c}{enough last year} \\
    \cmidrule[0.2pt](l){2-2}\cmidrule[0.2pt](l){3-3}\cmidrule[0.2pt](l){4-4}\cmidrule[0.2pt](l){5-5}
    & (1)& (2)& (3)& (4)  \\   \hline
& &   \\
White non-Hispanic (all)&       0.075\sym{**} &       0.055\sym{*}  &       0.016         &      -0.066\sym{*}  \\
                    &     (0.032)         &     (0.032)         &     (0.018)         &     (0.036)         \\
Observations        &        7252         &        7251         &        7251         &        7258         \\

& &   \\
Black non-Hispanic (all)&       0.109         &       0.012         &       0.058         &      -0.072         \\
                    &     (0.158)         &     (0.097)         &     (0.079)         &     (0.108)         \\
Observations        &         819         &         819         &         819         &         818         \\

& &   \\
Hispanic (all)      &       0.287\sym{*}  &       0.142         &       0.178         &       0.056         \\
                    &     (0.155)         &     (0.156)         &     (0.165)         &     (0.112)         \\
Observations        &        1035         &        1035         &        1036         &        1036         \\

& &   \\
Black or Hispanic (all)&       0.214\sym{*}  &       0.086         &       0.125         &      -0.002         \\
                    &     (0.113)         &     (0.087)         &     (0.087)         &     (0.101)         \\
Observations        &        1854         &        1854         &        1855         &        1854         \\

& &   \\
Non-White (all)     &       0.168\sym{*}  &       0.096         &       0.163\sym{**} &      -0.017         \\
                    &     (0.087)         &     (0.061)         &     (0.064)         &     (0.081)         \\
Observations        &        2515         &        2514         &        2514         &        2514         \\

& &   \\
\hline \hline
\multicolumn{5}{p{15.3cm}}{\scriptsize{\textbf{Notes:} All columns report the fuzzy difference-in-discontinuities estimates using data from the Northeast, Midwest, and West regions, and compare outcomes in 2012 and 2014. The models include linear controls for age, fully interacted with dummies for age 65 or older and 2014. Other controls in these models include indicators for gender, race/ethnicity, education and region. Samples for the regression models only include people between the ages of 55 and 75. Standard errors (in parentheses) are clustered by quarter of age.} } \\
\end{tabular}
}
\end{center}
\end{table}

\clearpage
% I would suggest adding commas to your figures for number of observations. I think this would make them easier to read.
% Your use of "non-Hispanic" in the table is strange to me. It implies that there are white Hispanic people as well as black Hispanic people, which I was not aware was possible.

%%%%%%%%%%%%%%%%%%%%%%%%%%%%%%%%%%%%%%%%%%%%%%%%%

%%%%%%%%%%%%%%%%%%%%%%%%%%%%%%%%%%%%%%%%%%%%%%%%%

\begin{table}
\begin{center}
{
\renewcommand{\arraystretch}{0.6}
\setlength{\tabcolsep}{3pt}
\captionsetup{font={normalsize,bf}}
\caption {Access to Care in Midwest and South Regions (Placebo Test)}  \label{placebo1}
\footnotesize
\centering
\begin{tabular}{lcccccc}
\hline\hline
& &   \\
&\multicolumn{4}{c}{Panel A: Baseline measures}   \\ \cmidrule[0.2pt](l){2-5}
    & \multicolumn{1}{c}{Delayed care} & \multicolumn{1}{c}{Did not get}& \multicolumn{1}{c}{Saw doctor}& \multicolumn{1}{c}{Hospital stay} \\
       & \multicolumn{1}{c}{last year} & \multicolumn{1}{c}{care last year}& \multicolumn{1}{c}{last year}& \multicolumn{1}{c}{last year} \\
    \cmidrule[0.2pt](l){2-2}\cmidrule[0.2pt](l){3-3}\cmidrule[0.2pt](l){4-4}\cmidrule[0.2pt](l){5-5}
& (1)& (2)& (3)& (4)  \\   \hline
& &   \\
Overall sample      &       0.003         &      -0.004         &       0.071         &       0.023         \\
                    &     (0.022)         &     (0.017)         &     (0.046)         &     (0.026)         \\
Observations        &       22622         &       22622         &        9377         &       22597         \\

& &   \\
\hline
& &   \\
&\multicolumn{4}{c}{Panel B: Alternative measures}   \\ \cmidrule[0.2pt](l){2-5}
        & \multicolumn{1}{c}{Could not} & \multicolumn{1}{c}{Could not}& \multicolumn{1}{c}{Could not}& \multicolumn{1}{c}{Could not get} \\
    & \multicolumn{1}{c}{afford prescription} & \multicolumn{1}{c}{afford to see a}& \multicolumn{1}{c}{afford follow-up}& \multicolumn{1}{c}{appointment soon} \\
       & \multicolumn{1}{c}{medicine last year} & \multicolumn{1}{c}{specialist last year}& \multicolumn{1}{c}{care last year}& \multicolumn{1}{c}{enough last year} \\
    \cmidrule[0.2pt](l){2-2}\cmidrule[0.2pt](l){3-3}\cmidrule[0.2pt](l){4-4}\cmidrule[0.2pt](l){5-5}
    & (1)& (2)& (3)& (4)  \\   \hline
& &   \\
Overall sample      &      -0.033         &      -0.033         &      -0.027         &      -0.036         \\
                    &     (0.031)         &     (0.026)         &     (0.019)         &     (0.025)         \\
Observations        &       11349         &       11343         &       11345         &       11351         \\

& &   \\
\hline \hline
\multicolumn{5}{p{14cm}}{\scriptsize{\textbf{Notes:} All columns report the fuzzy difference-in-discontinuities estimates using data from the Midwest and South regions, and compare outcomes in 2012 and 2014. The models include linear controls for age, fully interacted with dummies for age 65 or older and 2014. Other controls in these models include indicators for gender, race/ethnicity, education and region. Samples for the regression models only include people between the ages of 55 and 75. Standard errors (in parentheses) are clustered by quarter of age.} } \\
\end{tabular}
}
\end{center}
\end{table}

\begin{table}
\begin{center}
{
\renewcommand{\arraystretch}{0.5}
\setlength{\tabcolsep}{3pt}
\captionsetup{font={normalsize,bf}}
\caption {Access to Care between 2009 and 2013 (Placebo Test)}  \label{placebo3}
\footnotesize
\centering
\begin{tabular}{lcccccc}
\hline\hline
& &   \\
&\multicolumn{4}{c}{Panel A: Baseline measures}   \\ \cmidrule[0.2pt](l){2-5}
    & \multicolumn{1}{c}{Delayed care} & \multicolumn{1}{c}{Did not get}& \multicolumn{1}{c}{Saw doctor}& \multicolumn{1}{c}{Hospital stay} \\
       & \multicolumn{1}{c}{last year} & \multicolumn{1}{c}{care last year}& \multicolumn{1}{c}{last year}& \multicolumn{1}{c}{last year} \\
    \cmidrule[0.2pt](l){2-2}\cmidrule[0.2pt](l){3-3}\cmidrule[0.2pt](l){4-4}\cmidrule[0.2pt](l){5-5}
& (1)& (2)& (3)& (4)  \\   \hline
& &   \\
\multicolumn{4}{l}{\underline{2012 vs. 2013:}} \\
& &   \\
Overall sample      &      -0.011         &      -0.010         &       0.008         &      -0.028         \\
                    &     (0.020)         &     (0.017)         &     (0.033)         &     (0.033)         \\
Observations        &       23719         &       23719         &        9649         &       23694         \\

& &   \\
\multicolumn{4}{l}{\underline{2011 vs. 2012:}} \\
& &   \\
Overall sample      &      -0.015         &      -0.023         &       0.005         &      -0.007         \\
                    &     (0.023)         &     (0.018)         &     (0.035)         &     (0.036)         \\
Observations        &       22766         &       22766         &        9230         &       22750         \\

& &   \\
\multicolumn{4}{l}{\underline{2010 vs. 2011:}} \\
& &   \\
Overall sample      &       0.010         &       0.012         &       0.008         &      -0.013         \\
                    &     (0.027)         &     (0.017)         &     (0.042)         &     (0.030)         \\
Observations        &       19742         &       19742         &        7787         &       19728         \\

& &   \\
\multicolumn{4}{l}{\underline{2009 vs. 2010:}} \\
& &   \\
Overall sample      &      -0.014         &       0.014         &      -0.017         &       0.003         \\
                    &     (0.027)         &     (0.018)         &     (0.045)         &     (0.028)         \\
Observations        &       17612         &       17612         &        6854         &       17594         \\

& &   \\
\hline
& &   \\
&\multicolumn{4}{c}{Panel B: Alternative measures}   \\ \cmidrule[0.2pt](l){2-5}
        & \multicolumn{1}{c}{Could not} & \multicolumn{1}{c}{Could not}& \multicolumn{1}{c}{Could not}& \multicolumn{1}{c}{Could not get} \\
    & \multicolumn{1}{c}{afford prescription} & \multicolumn{1}{c}{afford to see a}& \multicolumn{1}{c}{afford follow-up}& \multicolumn{1}{c}{appointment soon} \\
       & \multicolumn{1}{c}{medicine last year} & \multicolumn{1}{c}{specialist last year}& \multicolumn{1}{c}{care last year}& \multicolumn{1}{c}{enough last year} \\
    \cmidrule[0.2pt](l){2-2}\cmidrule[0.2pt](l){3-3}\cmidrule[0.2pt](l){4-4}\cmidrule[0.2pt](l){5-5}
    & (1)& (2)& (3)& (4)  \\   \hline
 & &   \\
\multicolumn{4}{l}{\underline{2012 vs. 2013:}} \\
& &   \\
Overall sample      &       0.021         &      -0.002         &       0.025         &      -0.038\sym{*}  \\
                    &     (0.027)         &     (0.020)         &     (0.020)         &     (0.020)         \\
Observations        &       11700         &       11698         &       11695         &       11707         \\

& &   \\
\multicolumn{4}{l}{\underline{2011 vs. 2012:}} \\
& &   \\
Overall sample      &      -0.016         &      -0.008         &      -0.002         &       0.013         \\
                    &     (0.023)         &     (0.020)         &     (0.017)         &     (0.027)         \\
Observations        &       11201         &       11193         &       11194         &       11204         \\

& &   \\
\multicolumn{4}{l}{\underline{2010 vs. 2011:}} \\
& &   \\
Overall sample      &      -0.043         &      -0.040\sym{**} &               &       0.006         \\
                    &     (0.027)         &     (0.015)         &               &     (0.028)         \\
Observations        &        9573         &        5372         &                 &        9571         \\

    & &   \\
\multicolumn{4}{l}{\underline{2009 vs. 2010:}} \\
& &   \\
Overall sample      &       0.055    &     &    &   0.002         \\
                    &     (0.037)    & &     &     (0.028)         \\
Observations        &        8502    &&     &        8500         \\

& &   \\
\hline \hline
\multicolumn{5}{p{14cm}}{\scriptsize{\textbf{Notes:} All columns report the fuzzy difference-in-discontinuities estimates using data from the Northeast, Midwest, and West regions, and compare outcomes between 2009 and 2010, 2010 and 2011, 2011 and 2012, and 2012 and 2013. The models include linear controls for age, fully interacted with dummies for age 65 or older and 2014. Other controls in these models include indicators for gender, race/ethnicity, education and region. Samples for the regression models only include people between the ages of 55 and 75. Standard errors (in parentheses) are clustered by quarter of age.} } \\
\end{tabular}
}
\end{center}
\end{table}
% I would suggest adding commas to your figures for number of observations. I think this would make them easier to read.

%%%%%%%%%%%%%%%%%%%%%%%%%%%%%%%%%%%%%%%%%%%%%%%%%%

%%%%%%%%%%%%%%%%%%%%%%%%%%%%%%%%%%%%%%%%%%%%%%%%%%
%%ADDTIONAL ROBUSTNESS

%%%%%%%%%%%%%%%%%%%%%%%%%%%%%%%%%%%%%%%%%%%%%%%%%

\begin{table}
\begin{center}
{
\renewcommand{\arraystretch}{0.6}
\setlength{\tabcolsep}{3pt}
\captionsetup{font={normalsize,bf}}
\caption {Access to Care (Excluding Individuals who Turn 65 in the First Half of 2014)}  \label{access_no65}
\footnotesize
\centering
\begin{tabular}{lcccccc}
\hline\hline
& &   \\
&\multicolumn{4}{c}{Panel A: Baseline measures}   \\ \cmidrule[0.2pt](l){2-5}
    & \multicolumn{1}{c}{Delayed care} & \multicolumn{1}{c}{Did not get}& \multicolumn{1}{c}{Saw doctor}& \multicolumn{1}{c}{Hospital stay} \\
       & \multicolumn{1}{c}{last year} & \multicolumn{1}{c}{care last year}& \multicolumn{1}{c}{last year}& \multicolumn{1}{c}{last year} \\
    \cmidrule[0.2pt](l){2-2}\cmidrule[0.2pt](l){3-3}\cmidrule[0.2pt](l){4-4}\cmidrule[0.2pt](l){5-5}
& (1)& (2)& (3)& (4)  \\   \hline
& &   \\
Overall sample      &       0.022         &       0.011         &       0.063         &       0.042         \\
                    &     (0.020)         &     (0.015)         &     (0.040)         &     (0.027)         \\
Observations        &       25193         &       25193         &       10297         &       25167         \\

& &   \\
\hline
& &   \\
&\multicolumn{4}{c}{Panel B: Alternative measures}   \\ \cmidrule[0.2pt](l){2-5}
        & \multicolumn{1}{c}{Could not} & \multicolumn{1}{c}{Could not}& \multicolumn{1}{c}{Could not}& \multicolumn{1}{c}{Could not get} \\
    & \multicolumn{1}{c}{afford prescription} & \multicolumn{1}{c}{afford to see a}& \multicolumn{1}{c}{afford follow-up}& \multicolumn{1}{c}{appointment soon} \\
       & \multicolumn{1}{c}{medicine last year} & \multicolumn{1}{c}{specialist last year}& \multicolumn{1}{c}{care last year}& \multicolumn{1}{c}{enough last year} \\
    \cmidrule[0.2pt](l){2-2}\cmidrule[0.2pt](l){3-3}\cmidrule[0.2pt](l){4-4}\cmidrule[0.2pt](l){5-5}
    & (1)& (2)& (3)& (4)  \\   \hline
& &   \\
Overall sample      &       0.064\sym{**} &       0.068\sym{***}&       0.049\sym{***}&      -0.032         \\
                    &     (0.027)         &     (0.024)         &     (0.015)         &     (0.030)         \\
Observations        &       12404         &       12401         &       12401         &       12409         \\

& &   \\
\hline \hline
\multicolumn{5}{p{14cm}}{\scriptsize{\textbf{Notes:} All columns report the fuzzy difference-in-discontinuities estimates using data from the Northeast, Midwest, and West regions, and compare outcomes in 2012 and 2014. The models include linear controls for age, fully interacted with dummies for age 65 or older and 2014. Other controls in these models include indicators for gender, race/ethnicity, education and region. Samples for regression models only include people between the ages of 55 and 75, excluding individuals who turned 65 in the first half of 2014.  Standard errors (in parentheses) are clustered by quarter of age.} } \\
\end{tabular}
}
\end{center}
\end{table}
% I would suggest adding commas to your figures for number of observations. I think this would make them easier to read.

\end{document}